\documentclass{fundam}
\usepackage{url} 
\usepackage[ruled,lined]{algorithm2e}
\usepackage{graphicx}
\usepackage{tikz}
\usetikzlibrary{automata, positioning, arrows}

\usepackage{wrapfig,lipsum}
\usepackage{microtype}
\usepackage{color,soul}
\usepackage[colorinlistoftodos]{todonotes}
\usepackage{amssymb}
\usepackage{amsmath}
\usepackage{multirow}
\usepackage{ctable}
\usepackage{dsfont}

\AtBeginDocument{%
	\mathchardef\mathcomma\mathcode`\,
	\mathcode`\,="8000
}
{\catcode`,=\active
	\gdef,{\mathcomma\discretionary{}{}{}}
}
\mathchardef\breakingcomma\mathcode`\,
{\catcode`,=\active
	\gdef,{\breakingcomma\discretionary{}{}{}}
}

\newcommand\numberthis{\addtocounter{equation}{1}\tag{\theequation}}
\newcommand{\PNpre}[1]{{}^{\bullet}{#1}}    
\newcommand{\PNpost}[1]{{#1}^{\bullet}}     

  \usepackage{hyperref}

\begin{document}

%
\setcounter{page}{203}
\publyear{2021}
\papernumber{2087}
\volume{183}
\issue{3-4}

  \finalVersionForARXIV


\title{Inferring Unobserved Events in Systems with\\ Shared Resources and Queues}

\author{Dirk Fahland\thanks{Address  for correspondence: TU Eindhoven, PO Box 513, 5600MB Eindhoven, NL},  Vadim Denisov\\
Eindhoven University of Technology\\
Eindhoven, the Netherlands\\
\{d.fahland, v.denisov\}@tue.nl\vspace*{2mm}
\and Wil. M.P. van der Aalst\\
Process and Data Science (Informatik 9)\\
RWTH Aachen University, Aachen, Germany\\
wvdaalst@pads.rwth-aachen.de
 }

\maketitle

\runninghead{D. Fahland et al.}{Repairing Event Logs of Systems with Shared Resources}

\begin{abstract}
To identify the causes of performance problems or to predict process behavior, it is essential to have correct and complete event data. This is particularly important for distributed systems with shared resources, e.g., one case can block another case competing for the same machine, leading to inter-case dependencies in performance. However, due to a variety of reasons, real-life systems often record only a subset of all events taking place. To understand and analyze the behavior and performance of processes with shared resources, we aim to reconstruct bounds for timestamps of events in a case that must have happened but were not recorded by inference over events in other cases in the system. We formulate and solve the problem by systematically introducing multi-entity concepts in event logs and process models. We introduce a partial-order based model of a multi-entity event log and a corresponding compositional model for multi-entity processes. We define PQR-systems as a special class of multi-entity processes with shared resources and queues. We then study the problem of inferring from an incomplete event log unobserved events and their timestamps that are globally consistent with a PQR-system. We solve the problem by reconstructing unobserved traces of resources and queues according to the PQR-model and derive bounds for their timestamps using a linear program. While the problem is illustrated for  material handling systems like baggage handling systems in airports, the approach can be applied to other settings where recording is incomplete. The ideas have been implemented in ProM and were evaluated using both synthetic and real-life event logs.
\end{abstract}


\section{Introduction}
\label{sec:introduction}

Precise knowledge about actual process behavior and performance is required for identifying causes of performance issues~\cite{MarusterB_2009_redesigning}, as well as for predictive process monitoring of important process performance indicators~\cite{MarquezChamorro_Survey}. For Material Handling Systems (MHS), such as Baggage Handling Systems (BHS) of airports, performance incidents are usually investigated offline, using recorded event data  for finding root causes of problems~\cite{DFA_unbiased_fine_grained}, while online event streams are used as input for predictive performance models~\cite{AhmedPCL_OnlineRisk}. Both analysis and monitoring heavily rely on the completeness and accuracy of input data. For example, events may not be recorded and, as a result, we do not know when they happened even though we can derive that they must have happened. Yet, when different cases are competing for shared resources, it is important to reconstruct the ordering of events and provide bounds for non-observed timestamps.

\medskip
However, in most real-life systems, items are not continuously tracked and not all events are stored for cost-efficiency, leading to incomplete performance information which impedes precise analysis. For example, an MHS tracks the location of an item, e.g., a bag or box, via hardware sensors placed throughout the system, generating tracking events for system control, monitoring, analysis, and prediction. Historically, to reduce costs, a tracking sensor is only installed when it is strictly necessary for the correct execution of a particular operation, e.g., only for the precise positioning immediately before shifting a bag from one conveyor onto another. Moreover, even when a sensor is installed, an event still can be discarded to save storage space.
As a result, the recorded event data of an MHS are typically incomplete, hampering analysis based on such incomplete data. Therefore, it is essential to repair the event data before analysis.
Fig.~\ref{fig:ps} shows a simple MHS where events are not always recorded. The process model is given and for two cases the recorded incomplete sets of events are depicted using the so-called \emph{Performance Spectrum}~\cite{DFA_unbiased_fine_grained}.

\begin{figure}[t!]
\vspace*{-1mm}
    \includegraphics[width=\linewidth]{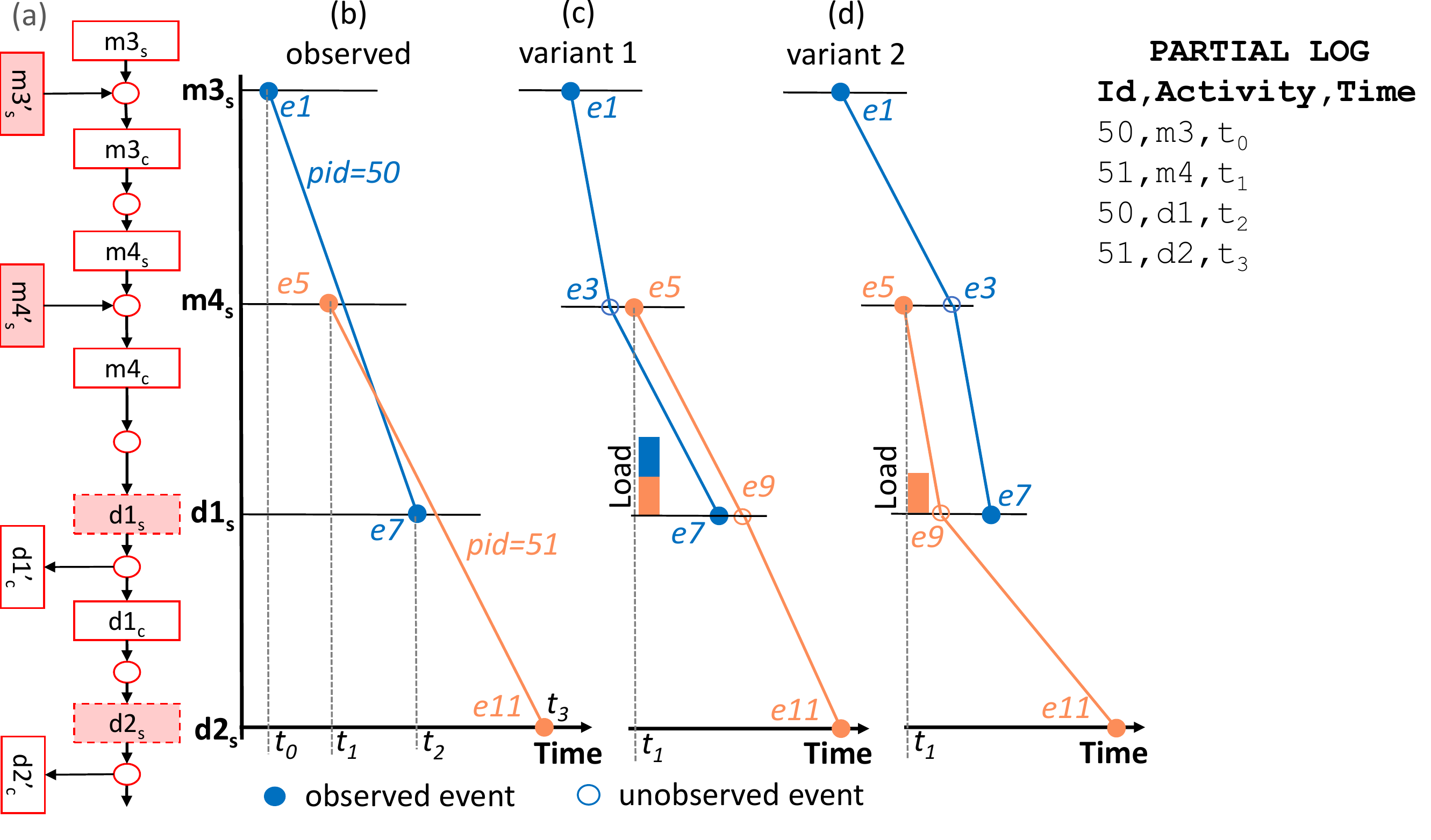}\vspace*{-3mm}
    \caption{An MHS model example (a), observed imprecise behavior for two cases 50 and 51 (b), possible actual behaviors (c,d).}
    \label{fig:ps}\vspace*{-2mm}
\end{figure}

\medskip
Fig.~\ref{fig:ps}(b) shows item pid=50 entering the system via $m3$ at time $t_0$ (event $e_1$) and leaving the system via $d1$ at time $t_2$ ($e_7$), and item pid=51 entering the system via $m4$ at time $t_1$ ($e_5$) and leaving the system via $d2$ at time $t_3$ ($e_{11}$). As only these four events are recorded, the event data do not provide information in which order both cases traversed the \emph{segment} $m4 \to d1$. Naively interpolating the movement of both items, as shown in Fig.~\ref{fig:ps}(b), suggests that item pid=51 overtakes item pid=50. This contradicts that all items are moved from $m4$ to $d1$ via a conveyor belt, i.e., a FIFO queue: item 51 cannot have overtaken item 50.
In contrast, Fig.~\ref{fig:ps}(c) and Fig.~\ref{fig:ps}(d) show two possible behaviors that are consistent with our knowledge of the system. We know that a conveyor belt (FIFO queue) is a shared resource between $m4$ and $d1$. Both variants differ in the order in which items 50 and 51 enter and leave the shared resource, the speed with which the resource operated, and the load and free capacity the resource had during this time.
In general, the longer the duration of naively interpolated segment occurrences, the larger the potential error. Errors in load, for example, make performance outlier analysis~\cite{DFA_unbiased_fine_grained} or short-term performance prediction~\cite{DFA_ICPM2019} rather difficult. Errors in order impede root-cause analysis of performance outliers, e.g., finding the cases that caused or were affected by outlier behavior.

\par\vspace{.25em}\noindent\textbf{Problem.} In this paper, we address a novel type of problem as illustrated in Fig.~\ref{fig:ps} and explained above. The behavior and performance of the system cannot be determined by the properties of each case in isolation, but depends on the \emph{behavior of other cases} and the \emph{behavior of the shared resources} involved in the cases.
Crucially, each case is handled by multiple resources and each resource handles multiple cases, resulting in \emph{many-to-many} relations between them.
The concrete problem we address is to reconstruct \emph{unobserved} behavior and performance information of \emph{each case} and each \emph{shared resource} in the system that is \emph{consistent} with both observed and reconstructed unobserved behavior and performance of all other cases and shared resources.
More specifically, we consider the following information as given: (1) an event log $L_1$ containing the case identifier, activity and time for recorded events where intermediate steps are not recorded (i.e., the event log may be incomplete), (2) a model of the process (i.e., possible paths for handling each individual case), and (3) a description (model) of the resources involved in each step (e.g., queues, single server resources and their performance parameters such as processing and waiting time). Based on the above input, we want to provide a complete event log $L_2$ that describes (1) for each case the exact sequence of process steps, (2) and for each unobserved event a time-window of earliest and latest occurrence of the event so that (3) either all earliest or all latest timestamps altogether describe a consistent execution of the entire process over all shared resources.

We elaborate the problem further in Sect.~\ref{sec:related_work} where we discuss related work. Specifically, prior work either only considers the case or the resource perspective explicitly, making implicit assumptions about their complex interplay. The goal of this paper is to explicitly account for the interplay of control-flow, resources, and queues in the entire \emph{system}. This requires us to first identify and develop suitable formal concepts that allow us to precisely state and automatically solve the above problem of inferring missing events and their time-stamps in a way that considers all perspectives jointly.

\medskip
\noindent\textbf{Contribution.} We approach the problem under the conceptual lens of treating each process case, resource, or queue as a separate entity exhibiting its own behavior. System behavior then is the result of multiple entities synchronizing in joint steps, e.g., when a resource starts working on a case. Section~\ref{sec:modelling} further develops the problem of inferring missing events under this conceptual lens.
To solve the problem, the paper systematically introduces multi-entity concepts in formal models for event logs and in process models with the following four contributions.

\smallskip
(1) To ground the problem in existing types of event data, we propose in Sect.~\ref{sec:logs} an alternative definition of event logs that can handle multiple entity identifiers. The information is carried in an event table with multiple entity identifier columns. We then show that the information in this table can be viewed from two different but equivalent perspectives: (i) as a family of sequential event logs, one per entity type; and (ii) as a global strict partial order over all events that is typed with entity types and can be understood as a \emph{system-level run}. This model allows us to conceptually decompose behavior (run of a system or event data) into individual \emph{entity traces} of process cases, resources, and queues. Different entity traces synchronize when a resource or queue is involved in a case, allowing to explicitly describe their many-to-many relations in the run.

\smallskip
(2) To provide a well-defined problem of repairing incomplete event logs, we develop a novel conceptual model for processes with shared resources in Sect.~\ref{sec:pqr}. We extend the recently proposed synchronous proclet model~\cite{DFahlandMultiDim} with concepts of coloured Petri nets~\cite{DBLP:journals/cacm/JensenK15} to precisely describe queueing and timed behavior in systems with multiple synchronizing entities, resulting in the model of \emph{CPN proclet systems}. We provide a replay semantics for CPN proclet systems that defines when a model accepts a given event log. Our semantics is compositional: the system can replay the log iff each component can replay the part of the log it relates to. A side product of this work is that we also provide a semantics for replaying event logs on regular coloured Petri nets.

\smallskip
(3) We then formalize a special class of CPN proclet systems called \emph{PQR systems} which are composed of one component for the process, and multiple components for shared resources and queues. PQR systems allow to model processes where each step is served by one single-server resource and resources are connected by strict FIFO queues only. These assumptions are reasonable for a large class of MHSs.

\smallskip
(4) We then provide an automated technique to solve the problem for PQR systems where the process is \emph{acyclic} which suffices for many real-life problem instances. The central solution idea given in Sect.~\ref{sec:solution} is to decompose the behavioral information in the incomplete event log into entity traces. We gradually infer unobserved events and unobserved entity traces and their synchronization with other entities from the component-based structure of the PQR system. We then formulate a Linear Programming (LP) problem~\cite{Schrijver86} to infer upper and lower bounds of timestamps of unobserved events based on bounds of timestamps along the different entity traces.

\medskip
We evaluated the approach by comparing the restored event logs with the ground truth for synthetic logs and estimate errors for real-life event logs for which the ground truth is unavailable (Sect.~\ref{sec:evaluation}).
We discuss our findings and alleys for future work in Sect.~\ref{sec:conclusion}.

\section{Related work}
\label{sec:related_work}

In all operational processes (logistics, manufacturing, healthcare, education and so on) complete and precise event data, including information about workload and resource utilization, is highly valuable since it allows for process mining techniques uncovering compliance and performance problems.
Event data can be used to replay processes on top of process models~\cite{Aalst_2016_pm_book}, to predict process behavior~\cite{Senderovich_DD,DFA_ICPM2019}, or to visualize detailed process behavior using performance spectra~\cite{DFA_unbiased_fine_grained}. All of these techniques rely on complete and correct event data. Since this is often not the case, we aim to transform \emph{incomplete} event data into \emph{complete} event data.

\medskip
Various approaches exist for dealing with incomplete data of processes with non-isolated cases that compete for scarce resources. In call-center processes, thoroughly studied in~\cite{GansCallCenters}, queueing theory models can be used for load predictions under assumptions about distributions of unobserved parameters, such as customer patience duration~\cite{BrownStatAnalysisCallCenter}, while assuming high load snapshot principle predictors show better accuracy~\cite{Senderovich2014QueueM}. For time predictions in congested systems, the required features are extracted using congestion graphs~\cite{SenderovichConGraphs} mined using queuing theory.

Techniques to repair, clean, and restore event data before analysis have been suggested in other works.
An extensive taxonomy of quality issue patterns in event logs is presented  in~\cite{SURIADI2017132}. The taxonomy specifically discusses how to detect and correct inadvertent time intervals (i.e., time stamps recorded later than the occurrence of the event) through domain knowledge; no automatic technique is provided. The timestamp repair technique in \cite{Conforti2018} automatically reconstructs the most likely order of wrongly recorded events and most likely intervals for timestamps based on other traces; the technique assumes all events were recorded and does not consider ordering constraints due to resources involved across traces. In \cite{MARTIN2020101463} resource availability calendars are retrieved from event logs without the use of a process model, but assuming $start$ and $complete$ life-cycle transitions as well as a case arrival time present in a log.
Using a process model, classical trace alignment algorithms~\cite{Carmona_ConfChecking} restore missing events but do not restore their timestamps. The authors conclude (see~\cite{Carmona_ConfChecking}, p. 262) that incorporating other dimensions, e.g., resources, for multi-perspective trace alignment and conformance checking is an important challenge for the near future.
Recently, also techniques for process discovery and conformance checking over uncertain event data were presented~\cite{PegoraroUncertainData,PegoraroDiscovery}. The output of our approach can provide the input needed for these techniques.

Multiple recent works address behavioral models for behavior over multiple different entities in one-to-many and many-to-many relations. The model of proclets thereby defines one behavioral model (a Petri net) per entity. Entities interact asynchronously via message exchange~\cite{AalstPROCLETS} or synchronously via dynamic transition synchronization~\cite{DFahlandMultiDim}. Object-centric Petri nets~\cite{DBLP:journals/fuin/AalstB20} are a special class of coloured Petri nets~\cite{DBLP:journals/cacm/JensenK15} that are structured to model the flow and synchronization of different objects (or entities); they correspond to synchronous proclets~\cite{DFahlandMultiDim} where the synchronization has been materialized in the model structure. Catalog nets~\cite{DBLP:conf/bpm/GhilardiGMR20} approach the problem from the side databases and model entity behavior by describing database updates through transitions; entity synchronization is similar to synchronous proclets. Process structures~\cite{DBLP:conf/caise/SteinauAR19} integrate relational modeling and behavioral modeling but are using dedicated behavioral model without existing analysis techniques. None of these works so far considered system-level entities such as queues and resources as part of the model to study how system-level entities impact process behavior. Further, none of these works has provided any techniques for reasoning about missing temporal and behavioral information across different entities.

Also data models for event data over multiple entities have been studied extensively in three forms. One type of event logs describe entities just as a sequence (or collection) of events~\cite{DBLP:journals/ijcis/PopovaFD15,DBLP:conf/sefm/Aalst19} where each event carries multiple entity identifier attributes, possibly even having multiple entity identifier values. Behavioral analysis requires to extract a trace per entity, thereby constructing a set of related sequential event logs~\cite{DBLP:journals/ijcis/PopovaFD15,DBLP:journals/tsc/LuNWF15}. Other works construct a partial order over all events using graphs: nodes are events, edges describe when two events directly precede/follow each other and are typed with the entity for which this relation was observed~\cite{DBLP:journals/tsc/WernerG15,DBLP:conf/bpm/EsserF19,DBLP:conf/simpda/BertiA19,DBLP:journals/jodsn/EsserF21}. In this paper, we show that the three representations are essentially equivalent and just materialize the data in different forms; reasoning about incomplete behavior across multiple entities benefits from being able to switch between these perspectives arbitrarily. We thereby adopt a more classical partial-order model instead of a graph as it simplifies reasoning.

Our work contributes to the problem of reconstructing behavior of cases and limited shared resources for which the cases compete.
We use the notion of \emph{proclets} first introduced in~\cite{AalstPROCLETS} and adapted for process mining in~\cite{DFahlandMultiDim} to approach the problem from control-flow and resource perspectives at once. We assume a system model given as a composition of a control-flow proclet (process) and resource/queue proclets. The given event log is a set of events with multiple entity identifiers. We restore missing events through classical trace alignments over control-flow proclets. The dynamic synchronization of proclets~\cite{DFahlandMultiDim} allows us to infer how and when sequential traces of resource entities must have traversed over the control-flow steps, which we express as a linear programming problem to compute time stamp intervals for the restored events. For the construction of the linear program we make extensive use of the partial ordering of events. Event logs repaired in this way enable the use of analysis which assume event logs to be complete.

Compared to a prior version of this article~\cite{DBLP:conf/apn/DenisovFA20}, we here provide a complete formalization of the problem and all underlying concepts, including the definitions of multi-entity event logs, CPN proclet systems and their replay semantics, and a formal definition of PQR systems.

\section{Modeling inter-case behavior via shared resources}
\label{sec:modelling}

Prior work (cf. Sect.~\ref{sec:related_work}) approaches the problem of analyzing the performance of systems with shared resources primarily either from the control-flow perspective~\cite{MARTIN2020101463,PegoraroUncertainData,PegoraroDiscovery,Senderovich_DD,DFA_ICPM2019} or the resource/queuing perspective~\cite{GansCallCenters,BrownStatAnalysisCallCenter,Senderovich2014QueueM,SenderovichConGraphs}, leading to information loss about the other perspective. In the following, we show how to conceptualize the problem from both perspectives at once using \emph{synchronous proclets}~\cite{DFahlandMultiDim} extended with a few concepts of coloured Petri Nets~\cite{DBLP:journals/cacm/JensenK15}. This way we are able to capture both control-flow and resource dynamics and their interaction as synchronizing entity traces. We introduce the model in Sect~\ref{sec:architecture} and use it to illustrate how incomplete logging incurs information loss for performance analysis in Sect.~\ref{sec:logging}.

\subsection{Processes-aware systems with shared resources}
\label{sec:architecture}
\label{sec:model}

We explain the dynamics of process-aware systems over shared resources using a BHS handling luggage. The process control-flow takes a bag from a source (e.g., check-in or transfer from another flight), to a destination (e.g., the airplane, transfer) along intermediate process steps (e.g., baggage scanning, storage). BHS resources are primarily single-server machines (e.g., baggage scanners) connected via conveyor belts, i.e., FIFO queues. Fig.~\ref{fig:mfd}(a) shows a typical system design pattern involving the control-flow and resource perspective: four parallel check-in desks (c1-c4) merge into one \emph{linear conveyor} through \emph{merge points} (m2-m4). \emph{Divert points} (d1 and d2) can route bags from the linear conveyor to \emph{scanners} (s1 and s2). Each merge point and scanner is preceded by a FIFO queue for buffering incoming cases (bags) in case the corresponding resource is busy. Fig.~\ref{fig:mfd}(b) shows the plain control-flow of this BHS (also called Material Flow Diagram (MFD)). A real-life BHS may contain hundreds of process steps and resources, and conveyors may also form loops. Each processing step in a BHS is served by a limited number of resources (in case of machines exactly one) with a minimum \emph{processing time} and often a minimum \emph{waiting time} to ensure sufficient ``operating space'' $p$ between two subsequent bags as shown in Fig.~\ref{fig:mfd}(a). Similarly, the conveyor belts realizing FIFO queues have certain operating speeds which determine a minimum \emph{waiting time} to reach the end of queue.

\begin{figure}[!h]
\vspace*{1mm}
    \centering\includegraphics[width=0.8\textwidth]{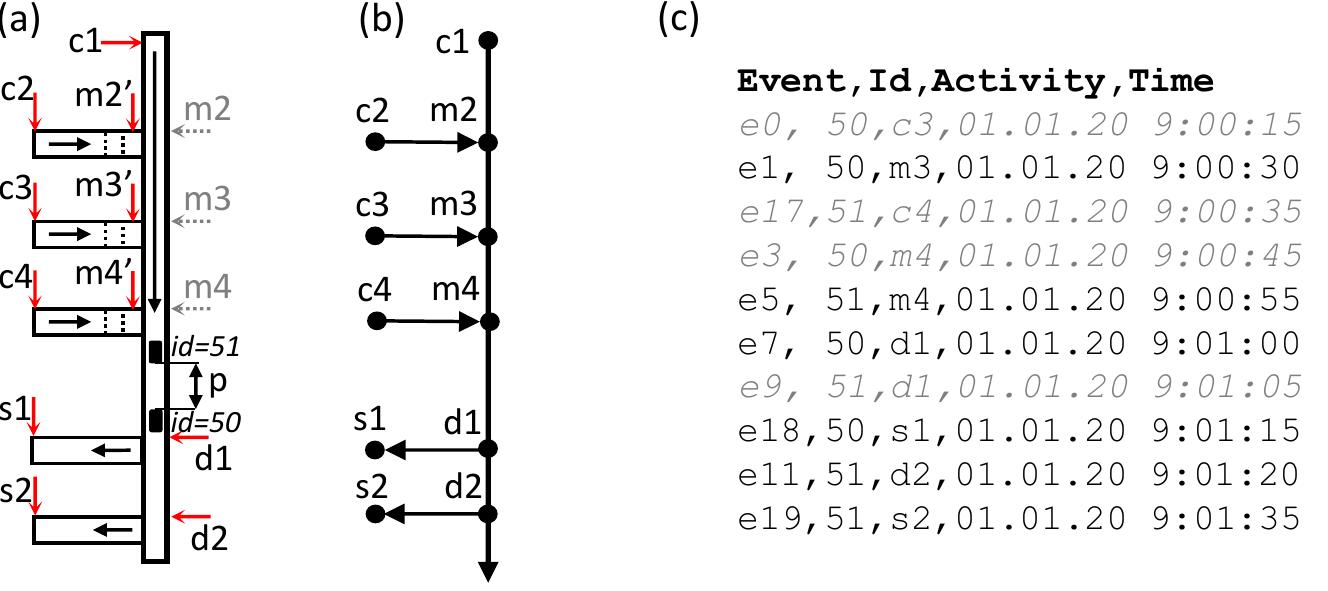}\vspace*{-1mm}
    \caption{A baggage handling system fragment (a) and its material flow diagram (b). Conveyor belts of check-in counters $c1-c4$ merge at points $m2-m4$, further downstream bags can divert at $d1$ and $d2$ to X-Ray security scanners $s1$ and $s2$. Red arrows show sensor (logging) locations. An example of an incomplete event log of the system in (a) is shown in (c), where missing events are shown in the grey color.}
    \label{fig:mfd}
\end{figure}

\begin{figure}
\vspace*{-2mm}
    \centering\includegraphics[width=0.7\textwidth]{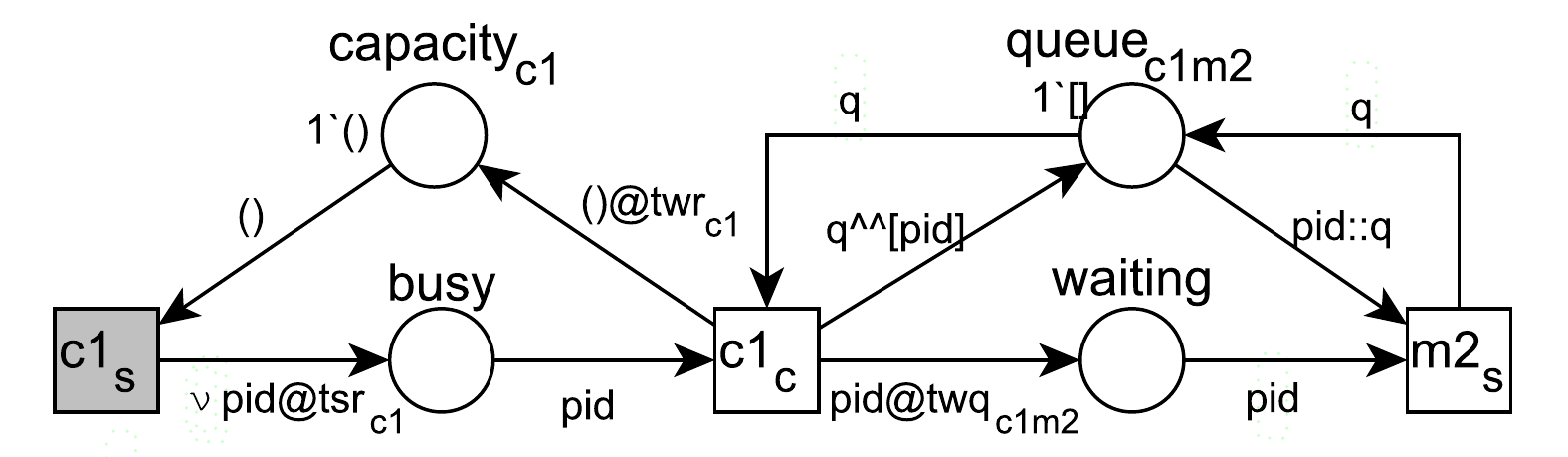}\vspace*{-1mm}
    \caption{Coloured Petri net model of conveyor $c1:m2$ of Fig~\ref{fig:mfd}.}
    \label{fig:cpn_conveyor}
\end{figure}

\medskip
\noindent\textbf{Modeling with Coloured Petri Nets.} Fig.~\ref{fig:cpn_conveyor} shows a coloured Petri Net (CPN) model for the segment from check-in $c1$ to merge step $m2$. In the model, transitions $c1_s$ and $c1_c$ describe \emph{start} and \emph{completion} of the check-in step $c1$. When $c1_s$ occurs, arc inscription $\nu\ pid$ produces a new identifier value for a bag (also called coloured token) on place \emph{busy} and the single token on place $\mathit{capacity}_{c1}$ is removed, i.e., no more resource is available to start $c1$ for another $pid$. By annotation $@tsr_{c1}$, the new $pid$ on $busy$ can only be consumed by transition $c1_c$ (to complete step $c1$) after service-time $tsr_{c1}$ has passed. When $c1_c$ occurs, a token is produced on $\mathit{capacity}_{c1}$ and waiting time $twr_{c1}$ has to pass before $c1_s$ can occur again. Also, the bag identified by $pid$ is removed from $\mathit{busy}$ and inserted into a FIFO queue (modeling a conveyor belt) between the end of $c1$ and the start of $m2$. Arc annotation $q\hat{\ }\hat{\ }[pid]$ specifies that bag $pid$ is appended to the end of queue $q$. Producing $pid$ with time annotation $@twq_{c1m2}$ on place $\mathit{waiting}$ models the minimum time it takes for a bag to travel from $c1$ to $m2$. Only then a bag may leave the queue at transition $m2_s$ where arc annotation $id\mathord{::}q$ specifies that bag $pid$ at the head is removed from the tail of the queue $q$.

The CPN model in Fig.~\ref{fig:cpn_conveyor} describes the impact of limited resource capacity and queues on the progress of a case, but does not model the resource and the queue as entities themselves. This makes it impossible to reason about resource and queue behavior explicitly. To alleviate this, we use proclets.

\medskip
\noindent\textbf{Modeling with Synchronous Proclets.}
A proclet is a Petri net that describes the behavior of a specific entity that can be distinguished through a unique identifier. Interactions between entities are described through synchronization channels between transitions of different proclets~\cite{DFahlandMultiDim}. The synchronous proclet system in Fig.~\ref{fig:bhs} describes the entire BHS of Fig.~\ref{fig:mfd}(a) by using three types of proclets.
\begin{enumerate}
\item The \emph{process proclet} (red dotted border) is a Petri net describing the control-flow perspective of how bags, identified by variable $pid$ may move through the system. It directly corresponds to the MFD of Fig.~\ref{fig:mfd}(b). It is transition-bordered and each occurrence of one of its initial transitions creates a new case identifier (a new value bound to variable $pid$) that was never seen before in the sense of $\nu$-Petri nets~\cite{VelardoF:2008:nu_nets}, see~\cite{DFahlandMultiDim} for details.
\item Each \emph{resource proclet} (green dashed border) models a resource with cyclic behavior as its own entity identified by variable $rid$. For example, the \emph{PassengerToSystemHandover} proclet (top left) identifies a concrete resource by token $rid = c1$; its life-cycle models that starting a task ($c1_s$) makes the resource \emph{busy} and takes service time $tsr_{c1}$, after completing the task ($c1_c$) the resource has waiting time $twr_{c1}$ before being \emph{idle} again in the same way as Fig.~\ref{fig:cpn_conveyor}. Which item the resource is busy with is recorded through variable $pid$ in the pair $(rid,pid)$. In the classical CPN in Fig.~\ref{fig:cpn_conveyor}, the $pid$ is determined by the input and output transition of place \emph{busy}. For the resource proclets in Fig.~\ref{fig:bhs}, \emph{pid} is a free variable at $c1_s$ and $c1_c$ whose value is determined when synchronizing with the corresponding transition in the Process proclet (which we describe below). All other resource proclets follow the same pattern, though some resources such as \emph{MergingUnit-m2} and \emph{DivertingUnit-d1} may have two transitions to become \emph{busy} or \emph{idle}, respectively.
\item Each queue proclet (blue dash-dotted border) describes a FIFO queue as in Fig.~\ref{fig:cpn_conveyor} from the end transition of one task to the start transition of the subsequent task, e.g., from $c1_c$ (end of $c1$) to $m2_s$ (start of $m2$). However, where Fig.~\ref{fig:cpn_conveyor} uses a distinct place $\mathit{queue}_{c1m2}$ for the queue, a queue proclet maintains the queue state (the list) together with the queue identifier $qid$ in place $\mathit{queue}$. Items entering the queue are remember by their $pid$. In the classical CPN in Fig.~\ref{fig:cpn_conveyor}, the $pid$ is determined by consuming from the input place of transition $c1_c$. For the queue proclets in Fig.~\ref{fig:bhs}, \emph{pid} is a free variable at $c1_c$ whose value is determined when synchronizing with the corresponding transition in the Process proclet.
\end{enumerate}
Where the model of Fig.~\ref{fig:cpn_conveyor} only uses identifiers for $pid$ and distinguishes resource and queue through model structure, the resource and queue proclets explicitly model resource and queue identifiers through markings and variables. This will later allow us to relate event data over multiple identifiers to a proclet model and to decompose analysis problems along identifiers.

\medskip
\noindent{}\textbf{Transition Synchronization in Proclets.}
The proclet system synchronizes process, resources, and queues via \emph{synchronous channels} between transitions. A transition linked to a synchronous channel may only occur when all linked transitions are enabled; when they occur, they occur in a single synchronized event. For example, transition $c1_s$ is always enabled in \emph{Process}, generating a new bag id, e.g., $pid=49$, but it may only occur together with $c1_s$ in \emph{PassengerToSystemHandover}, i.e., when resource $c1$ is \emph{idle}, thereby synchronizing the process case for bag $pid=49$ with the resource with identifier $rid = c1$. By storing the pair $(rid,pid) = (c1,49)$ on place \emph{busy}, resource $c1$ is now \emph{correlated} to case $49$. Transition $c1_c$ of the process proclet can now only occur when synchronizing with $c1_c$ of the resource proclet, and thus only for $pid=49$ and $rid=c1$. Moreover, both $c1_c$ transitions can only occur when synchronizing with the $c1_c$ transition of the queue proclet for $qid=c1:m2$, thereby completing the work of $c1$ on $pid=49$ and putting $pid=49$ into the queue. Note that using CPN expressions (as used in queue and resource proclets) eliminates the need for dedicated correlation expressions used for the basic proclet model introduced in~\cite{DFahlandMultiDim}.

\begin{figure}[htbp]
\vspace*{-1mm}
    \includegraphics[height=0.9\textheight]{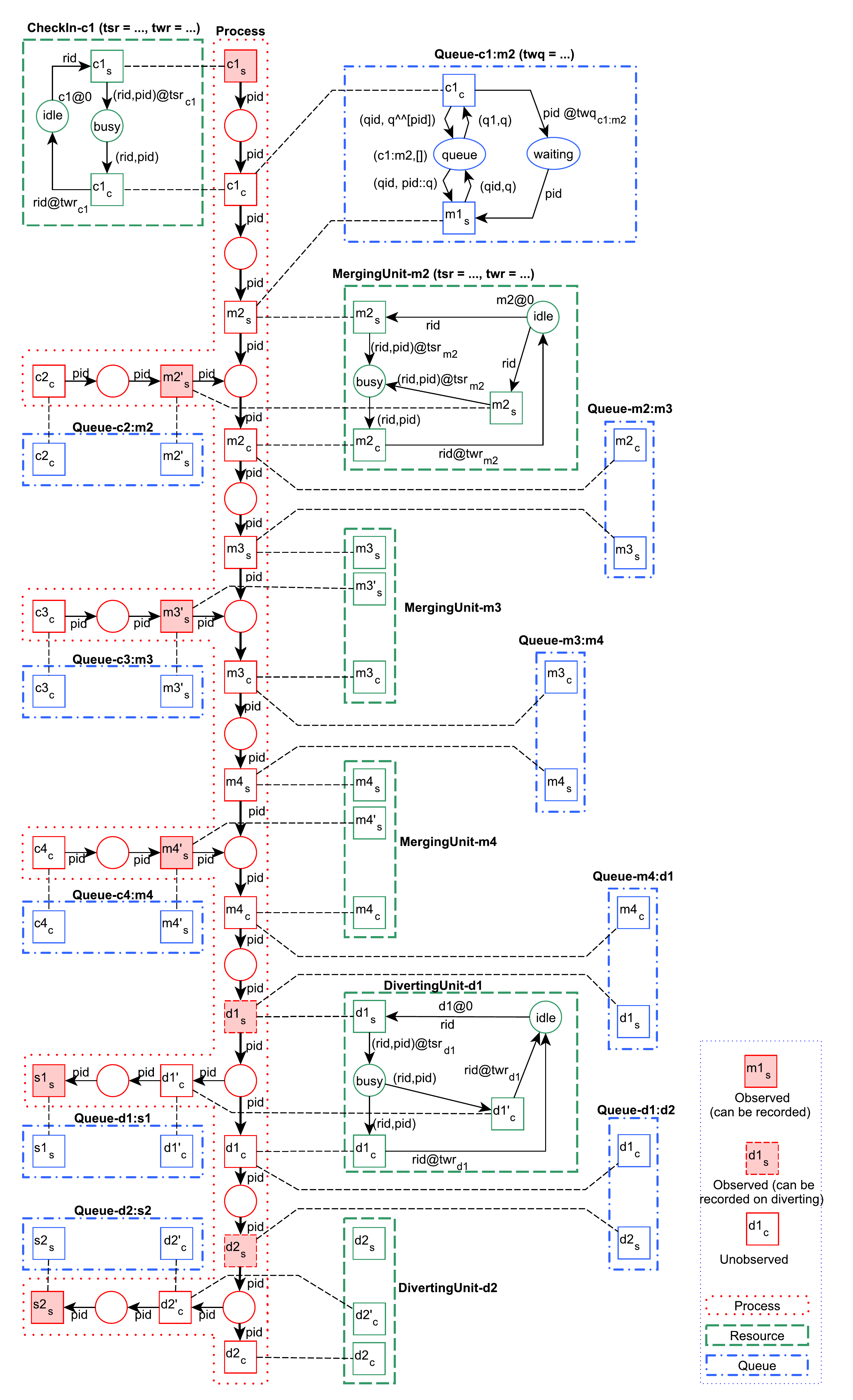}\vspace*{-1mm}
    \caption{The synchronous proclet model of the system shown in Fig.~\ref{fig:mfd}(a) consists of three types of proclets: \textit{Process} for modeling a system layout and process control flow (red, dotted), \textit{Resource} for modeling equipment performing tasks (green, dashed), and \textit{Queue} for modeling conveyors transporting bags in the FIFO order (blue dash-dotted). Only filled transitions can be observed in an event log.}
    \label{fig:bhs}
\end{figure}

\begin{figure}[htbp]
\vspace*{-1mm}
    \includegraphics[height=0.9\textheight]{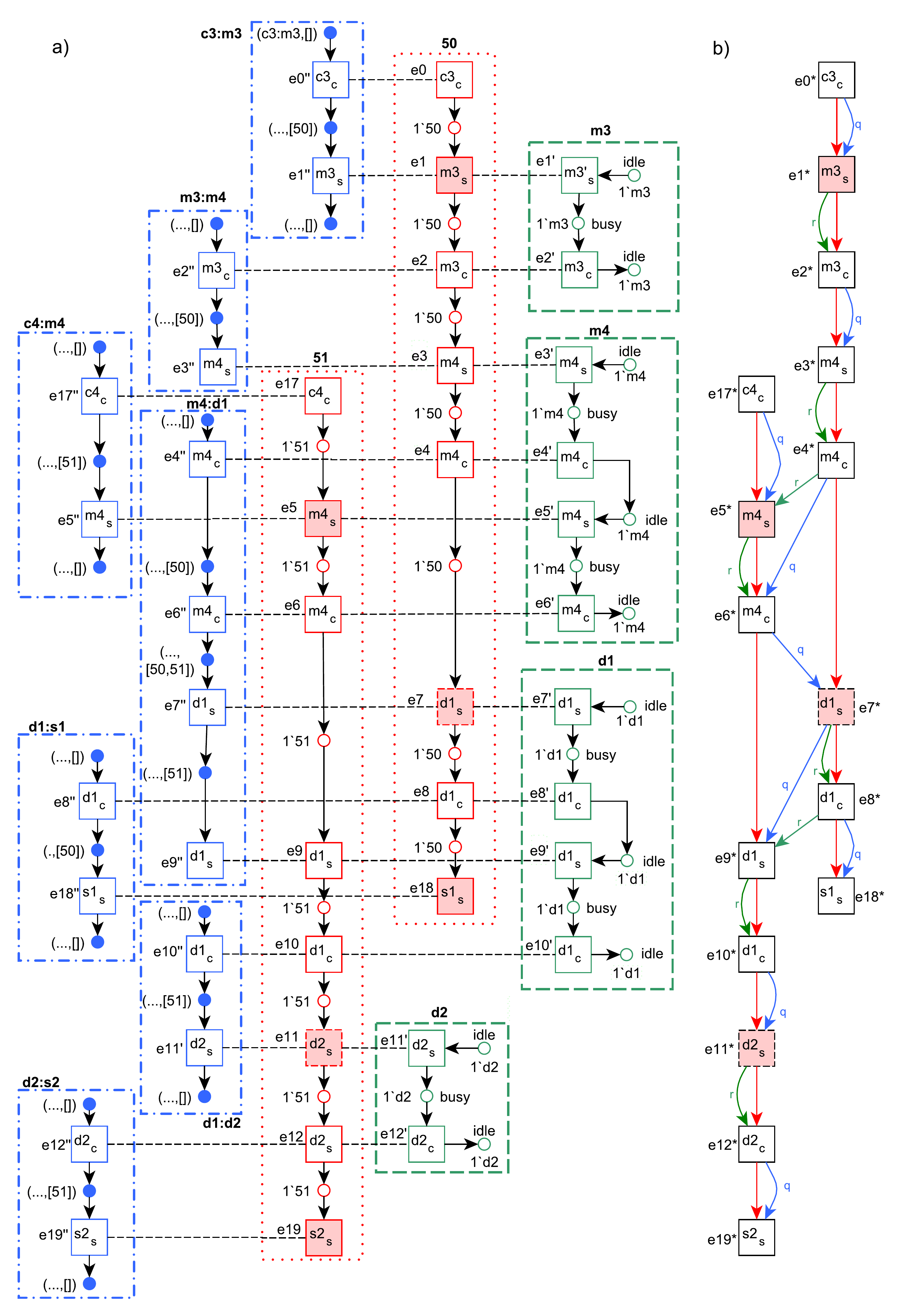}\vspace*{-1mm}
    \caption{Synchronization of multiple sub-runs of the synchronous proclet system in Fig.~\ref{fig:bhs} over shared resources and queues (a), and a global partial order obtained by the union of partial orders of each sub-run (b) for synchronized events, shown by red arrows, green arrows (with marker r) and blue arrows (with marker q) for partial orders $<_{pid}$, $<_{rid}$ and $<_{qid}$ respectively.}
    \label{fig:example}
\end{figure}

In the example, each resource is statically linked to one process step, but the model also allows for one resource to participate in multiple different process steps, and multiple resources to be required for one process step. In the following, we call a proclet system that defines proclets for processes, queues, and resource that are linked via synchronous channels as described above, a \emph{PQR system}.

\medskip
\noindent{}\textbf{Proclets Describe Synchronizing Entity Traces.} We now highlight how the partial-order semantics of synchronous proclets~\cite{DFahlandMultiDim} preserves the identities of process, resources, and queues as ``entity traces''. Figure~\ref{fig:example}(b) shows a partially-ordered run of the PQR system of Fig.~\ref{fig:bhs} for two bags $id=50$ and $id=51$. The run in Fig.~\ref{fig:example}(b) can be understood as a synchronization of multiple runs or traces of the process, resource, and queue proclets, one for each case, resource, or queue involved as shown in Fig.~\ref{fig:example}(a).

Bag $50$ gets inserted via input transition $c3_c$ (event $e_0^*$ in Fig.~\ref{fig:example}(b)). This event is a \emph{synchronization} of events $e0$ ($c3_c$ occurs for bag $50$ in the \emph{Process} proclet) and $e0'$ ($c3_c$ occurs for the \emph{c3:m3} queue) in Fig.~\ref{fig:example}(a). The minimal waiting time $twq_{c3m3}$ must pass before bag $50$ reaches the end of the queue and process step $m3$ can start. The process step $m3$ merges bag $50$ from the check-in conveyor $c3$ onto the main linear conveyor and may only start via transition $m3_s$ when \emph{MergingUnit-m3} is \emph{idle}. As this is the case, bag $50$ leaves the queue  ($e1''$ in \emph{c3:m3}), $m3$ starts merging ($e1'$ in \emph{m3}), the bag starts the merging step (event $e1$ in \emph{Process}), resulting in the synchronized event $e1^*$ in Fig.~\ref{fig:example}(b).

\eject

By $e_1'$, resource \emph{m3} switches from $idle$ to $busy$ and takes time $tsr_{m3}$ before it can complete the merge step with $m3_c$ (event $e2'$) on bag $50$ (event $e2)$; this merge step also inserts bag $50$ into queue \emph{m3:m4} ($e2''$) resulting in synchronized event $e2^*$. Subsequently, bag \emph{50} leaves queue \emph{m3:m4} ($e3^*$) is pushed by merge unit \emph{m4} into queue \emph{m4:d1} ($e4^*$).

\medskip
Concurrently, bag $51$ is inserted via input transition $c4_c$ (event $e17^*$), moves via queue \emph{c4:m4} also to merge unit \emph{m4} to enter queue \emph{m4:d1}, i.e., both bags $50$ and $51$ now compete for merge unit \emph{m4} and the order of entering \emph{m4:d1}. In the run in Fig.~\ref{fig:example}, \emph{m4} executes $m4_s$ and $m4_c$ for bag $51$ ($e5^*$ and $e6^*$) \emph{after completing} this step for bag $50$ ($e3^*$ and $e4^*$). Thus, $51$ enters the queue ($e6^*$) after $50$ entered the queue ($e5^*$) but before $50$ leaves the queue $e7^*$. Consequently, divert unit \emph{d1} first serves $50$ ($e7^*$ and $e8^*$) to reach scanner \emph{s1} ($e{18}^*$) before serving $51$ ($e9^*$ and $e{10}^*$) to reach scanner \emph{s2} ($e19^*$).

Fig.~\ref{fig:example}(b) shows how the process tokens of bag $50$ and $51$ synchronized with the resources and queue tokens along the run, forming sequences or \emph{traces} of events where each of these tokens was involved. For example, bag $50$ followed the trace $e0^*, e1^*,\ldots,e8^*,e{18}^*$ and queue \emph{m4:d1} followed trace $e4^*,e6^*, e7^*, e9^*$ thereby synchronizing with both bag $50$ and bag $51$.

\subsection{Information loss because of incomplete logging}
\label{sec:logging}

Although event data on objects that are tracked can be used for various kinds of data analysis~\cite{AhmedPCL_OnlineRisk,DFA_ICPM2019}, in practice sensors are placed only where it is absolutely necessary for correct operation of the system, e.g., for merge and divert operations, without considering data analysis needs. Applied to our example, only the transitions that are shaded in Fig.~\ref{fig:bhs} would be logged, i.e.,  $c1_s, m2_s', m3_s', m4_s', d1_s, d2_s, s1_s, s2_2$ would be logged from the \emph{control-flow} perspective only. The run of Fig.~\ref{fig:example} would result in a ``typical'' but highly incomplete event log as shown in Fig.~\ref{fig:mfd}(c).

According to this incomplete log, bag $50$ silently passes \emph{m4} and is tracked again only at \emph{d1} ($e7$) and finally at \emph{s1}  ($e{18}$) whereas $51$ silently passes \emph{d1} (as it moves further on the main conveyor) and is tracked again only at \emph{d2} ($e11$). Based on this incomplete information the bags $50$ and $51$ may have traversed \emph{m4:d1} in different orders and at different speeds resulting also in different loads as illustrated in Fig.~\ref{fig:ps}. As a result, in case of congestion, we cannot determine the ordering of cases~\cite{DFA_unbiased_fine_grained}, cannot compute the exact load on each conveyor part for (predictive) process monitoring~\cite{DFA_ICPM2019,Senderovich_DD}. The longer an unobserved path (e.g., $c1 \to d2$), the higher the uncertainty about the actual behavior and the less accurate performance analysis outcome.

Although minimal (or even average) service and waiting times on conveyor belts and resource are known, we need to determine the \emph{order} of all missing events and the possible \emph{intervals of their timestamps} to reconstruct for how long resources were occupied by particular cases and in which order cases were handled, e.g., did $50$ precede $51$ on \emph{m4:d1} or vice versa?

The objective of this paper is to reconstruct from a subset of events logged from the control-flow perspective only the remaining events (including time information), so that the time order is consistent with a partially ordered run of the entire system, including resource and queue proclets.
For example, from the recorded events of the event log in Fig.~\ref{fig:mfd}(c) we reconstruct the remaining events (Fig.~\ref{fig:example}(a)) with time information so that the resulting order (by time) is consistent with the partially ordered run in Fig.~\ref{fig:example}(b).

\section{Modeling system-level runs from event data}
\label{sec:logs}

Having introduced the problem in an informal way in Sect.~\ref{sec:modelling}, we now turn to formalizing it. We first discuss a behavioral model that describes, both, the behavior of all process executions in the system, and how these process executions interact via shared resource. To this end, we develop the notion \emph{system-level runs} and how they canonically emerge from classical event logs that also log the shared resources involved in the process execution.

\medskip
We first recall classical event logs in Sect.~\ref{sec:logs:classical} and their traces. We then formalize \emph{multi-entity event logs}  in Sect.~\ref{sec:logs:multi} where an event can be related to multiple different entities. This allows us to observe behavior along a specific shared resources in the same way as we observe behavior along a case identifiers.

\medskip
To later be able to reason about behavior along multiple entities we then introduce two different but equivalent views on a multi-entity event log that differ in how they explicate behavior over multiple entities.
\begin{itemize}
\item Each multi-entity event log induces a family of classical event logs (one per entity type) that synchronize on shared events; we introduce this view in Sect.~\ref{sec:logs:multi_sequential}.
\item Each multi-entity event log also induces a strict partial order over the events, where events are ordered over time along the same entity. We call this view a \emph{system-level run} and introduce it in Sect.~\ref{sec:logs:multi_spo}.
\end{itemize}
We show in Sect.~\ref{sec:logs:multi_sequential_spo} that all views contain the same information allowing us to switch perspectives on the behavior. We will use the different perspectives when formally stating the problem in Sect.~\ref{sec:pqr} and solving the problem in Sect.~\ref{sec:solution}.

\subsection{Classical event logs}
\label{sec:logs:classical}

We first provide a definition of a classical event log, which we call \emph{single entity event log}. We later generalize this definition to a \emph{multi-entity event log}. With this aim of generalization in mind, we define a single entity event log just as a set of events with attributes.  The cases and traces of an event log will then be derived from event attributes through canonical functions we provide afterwards.

From the usual event attributes of activity, time, and case identifier, only the activity name attribute \emph{act} is mandatory. The \emph{time} attribute is optional as we later want to study situations of incomplete logging. Also the case identifier attribute is optional for the same reason. When we later move to a multi-entity setting the term ``case'' is no longer adequate. We therefore call the case identifer attribute an \emph{entity type} attribute \emph{et}, referring to the type of entity on which the events are recorded (e.g., bags in baggage handling system)
\begin{definition}[Single entity event log]\label{def:event_log_single}
A \emph{single entity event log} $L = (E,\mathit{AN},\mathit{et},\#)$ is a set $E$ of events, a non-empty set $\mathit{AN}$ of attribute names with $time,act \in \mathit{AN}$ and a designated \emph{entity type attribute} $\mathit{et} \in AN$. The partial function $\# : E \times \mathit{AN} \nrightarrow \mathit{Val}$ assigns events $e \in E$ and attribute names $a \in \mathit{AN}$ a value $\#_{a}(e) = v$, so that the activity name $\#_{act}(e)$ is defined for each event $e \in E$.
\end{definition}
We write $\#_a(e) = \perp$ if event $e$ has no value defined for attribute $a$. We call the value $\#_{et}(e) = id$ the entity identifier of $e$ (for entity type $et$). Note that we do not require all events to be correlated to the designated entity type, i.e., $\#_{et}(e)$ can be undefined. Such events will later simply not be part of a case and trace. Events without time stamp $\#_{time}(e) = \perp$ are unordered to all other events. To distinguish event logs where all events are correlated to an entity and are ordered by time stamps, we introduce the following definitions.
\begin{definition}[time-incomplete, time-monotone event log]\label{def:event_log_time-incomplete}\label{def:time-monotone}
We call a single entity-event log $L = (E,\mathit{AN},\mathit{et},\#)$ \emph{time-complete} iff for each $e \in E$ holds $\#_{time}(e) \neq \perp$ and $\#_{et}(e) \neq \perp$, i.e., each event has activity, time, and entity type. Otherwise $L$ is called \emph{time-incomplete}. We call a complete log $L$ \emph{time-monotone} iff for any two $e,e'\in E$ holds if $\#_{et}(e) = \#_{et}(e')$ then $\#_{time}(e) \neq \#_{et}(e')$.
\end{definition}

\begin{table}[h]
\vspace*{-3mm}
\centering
\caption{Event log with multiple entity types pid, rid, qid}\label{tab:event_table_multiid}
{\small\sffamily
\begin{tabular}{lrrrrr}
  \hline
  event id & pid & activity & time & rid & qid \\
  \hline
$e_0$ & 50 & $c3_c$ & 01.01.20 9:00:15 & $\perp$ & c3:m3\\
$e_1$ & 50 & $m3_s$ & 01.01.20 9:00:30 & m3 & c3:m3\\
$e_2$ & 50 & $m3_c$ & 01.01.20 9:00:40 & m3 & m3:m4\\
$e_3$ & 50 & $m4_s$ & 01.01.20 9:00:45 & m4 & m3:m4\\
$e_4$ & 50 & $m4_c$ & 01.01.20 9:00:50 & m4 & m4:d1\\
$e_7$ & 50 & $d1_s$ & 01.01.20 9:01:05 & d1 & m4:d1\\
$e_8$ & 50 & $d1_c$ & 01.01.20 9:01:10 & d1 & d1:s1\\
$e_{18}$ & 50 & $s1_s$ & 01.01.20 9:01:15 & $\perp$ & d1:s1\\
  \hline
$e_{17}$ & 51 & $c4_c$ & 01.01.20 9:00:35 & $\perp$ & c4:m4\\
$e_5$ & 51 & $m4_s$ & 01.01.20 9:00:55 & m4 & c4:m4\\
$e_6$ & 51 & $m4_c$ & 01.01.20 9:01:00 & m4 & m4:d1\\
$e_9$ & 51 & $d1_s$ & 01.01.20 9:01:15 & d1 & m4:d1\\
$e_{10}$ & 51 & $d1_c$ & 01.01.20 9:01:20 & d1 & d1:d2\\
$e_{11}$ & 51 & $d2_s$ & 01.01.20 9:01:25 & d2 & d1:d2\\
$e_{12}$ & 51 & $d2_c$ & 01.01.20 9:01:30 & d2 & d2:s2\\
$e_{19}$ & 51 & $s2_s$ & 01.01.20 9:01:35 & $\perp$ & d2:s2\\
\hline
\end{tabular}}\vspace*{1mm}
\end{table}

Table~\ref{tab:event_table_multiid} shows a single entity event log for entity type $pid$. The log is time-complete and monotone: no two events for the same entity type carry the same time-stamp.

All events with the same value for $\mathit{et}$ are correlated to the same entity or \emph{case}. A \emph{trace} is the sequence of all events in a case ordered by time (and events without time-stamp can be placed anywhere in the sequence).
\begin{definition}[Case, trace, sequential event log]\label{def:case}\label{def:trace}
Let $L = (E,\mathit{AN},\mathit{et},\#)$ be an event log with entity type attribute $\mathit{et}$.

The set of \emph{cases} in $L$ wrt.\ $\mathit{et}$ is $\mathit{et}(L) = \{ \#_{et}(e) \mid e \in E \}$, i.e., all entity (or case) identifier values in $L$.
\eject

All events carrying the same case identifier value $\mathit{id} \in \mathit{et}(L)$ are \emph{correlated to $\mathit{id}$}, i.e., $\mathit{corr}(L,\mathit{et}=\mathit{id}) = \{ e \in E \mid \#_{et}(e) = id \}$.

A \emph{trace} of case $id$ is a sequence $\langle e_1,\ldots,e_n \rangle$ of all events correlated to $\mathit{id}$ that preserves time, i.e., $\mathit{corr}(L,\mathit{et}=\mathit{id}) = \{e_1,\ldots,e_n\}$ and for all $1 \leq i < j \leq n$ hold if $\#_{time}(e_i) \neq \perp$ and $\#_{time}(e_j) \neq \perp$ then $\#_{time}(e_i) \leq \#_{time}(e_j)$. If two events have the same timestamp, a case $id$ has more than one trace; we write $\sigma(L,\mathit{et}=id)$ for the set of traces of case $id$.

A \emph{sequential event} log of $L$ is a set $\sigma(L,\mathit{et})$ which contains for each $\mathit{id} \in \mathit{et}(L)$ exactly one trace $\sigma \in \sigma(L,\mathit{et}=id)$.
\end{definition}
For a time-incomplete event log $L$, the notion of trace and sequential event log are non-deterministic, i.e., an event without time stamp can be placed at an arbitrary position in the trace, allowing for multiple different traces for the same case. Only in a monotone event log $L$, each case $id$ has a unique trace $\{ \sigma_{et}^{id} \} = \sigma(L,\mathit{et}=\mathit{id})$ and the log $\sigma(L)$ is uniquely defined. We then write $\sigma_{et}^{id} = \sigma(L,\mathit{et}=\mathit{id})$.

Table~\ref{tab:event_table_multiid} shows a time-complete, monotone single entity event log for entity type $pid$ defining cases $\{50,51\}$ and traces $\sigma(L,pid,50) = \sigma_{pid}^{50} = \langle e_0,e_1,\ldots,e_8,e_{18} \rangle$ and $\sigma_{pid}^{51} = \langle e_{17},e_5,e_6,e_9,\ldots,e_{12},e_{19} \rangle$. This classical interpretation of the events in Table~\ref{tab:event_table_multiid} describes how bags $50$ and $51$ travel through the baggage handling system of Fig.~\ref{fig:mfd}(a) from check-in $c3_c$ and $c4_c$ to scanners $s1_s$ and $s2_s$.

\subsection{Event logs over multiple entities}
\label{sec:logs:multi}

In the classical single-entity event logs of Sect.~\ref{sec:logs:classical}, attributes $rid$ and $qid$ of Table~\ref{tab:event_table_multiid} are considered so-called \emph{event attributes}~\cite{Aalst_2016_pm_book} which describe the event further, i.e., event $e_6$ of bag $51$ at $m4_c$ was performed by resource $\#_{rid}(e_6) = m4$ (merge-unit 4) as the bag entered the conveyor belt $\#_{qid}(e_6) = m4:d1$ from merge-unit 4 to divert-unit 1.

However, attributes $rid$ and $qid$ do refer to system entities in their own right: the machines that perform the various activities on the bags, and the conveyor belts that move bags between activities and machines. These machines and conveyor belts exist beyond individual cases and occur also in other cases, e.g., $\#_{rid}(e_3) = m4, \#_{rid}(e_4) = m4:d1$ for bag $50$. To study how these shared resources (e.g., machines, conveyor belts) relate and order bags over time, we introduce the notion of a \emph{multi-entity event log} which designates multiple entity type attributes.
\begin{definition}[Multi-entity event log]\label{def:event_log_multiid}
An \emph{event log with multiple entity types} $L = (E,\mathit{AN},\mathit{ET},\#)$ is a set $E$ of events, a non-empty set $\mathit{AN}$ of attribute names with $act \in \mathit{AN}$ and a subset $\emptyset \neq \mathit{ET} \subset \mathit{AN}$ is designated as \emph{entity types}. Partial function $\# : E \times \mathit{AN} \nrightarrow \mathit{Val}$ assignings events $e \in E$ and attribute names $a \in \mathit{AN}$ a value $\#_{a}(e) = v$, so that the activity name $\#_{act}(e)$ is defined for each event $e \in E$.
\end{definition}
Table~\ref{tab:event_table_multiid} shows a monotone event log with multiple entity types $\mathit{ET} = \{pid, rid, qid\}$. In contrast to a single-entity log (Def.~\ref{def:event_log_single}), an event in a multi-entity log (Def.~\ref{def:event_log_multiid}) may carry more than one entity type $\#_{et}(e) \neq\perp, et \in \mathit{ET}$, e.g., $\#_{pid}(e_0) = 50, \#_{qid}(e_0) = c3:m3$. As for Def.~\ref{def:event_log_single}, an entity type may be undefined or one or multiple events, e.g., $\#_{rid}(e_0) = \perp$.

Note that a multi-entity event log $L = (E,\mathit{AN},\mathit{ET},\#)$ with a singleton set of identifiers $\mathit{ET} = \{ \mathit{et} \}$ coincides with a classical event log (Def.~\ref{def:event_log_single}). The notions of time-incomplete, complete, and monotone event log lift to multi-entity event logs by applying them on all entity types in $\mathit{ET}$.

We next introduce two views that materialize the behavioral information in a multi-entity event log: as sets of sequential event logs and as a partial order.

\subsection{Sequential view on event logs over multiple entities}
\label{sec:logs:multi_sequential}

We use sequential traces (Def.~\ref{def:trace}) to describe the behavior stored in a single-entity event log $L$ in an explicit form. Each trace in $\sigma(L,id)$ describes a possible sequences of activity executions over time for entity $id \in \mathit{et}(L)$; no two traces $\sigma(L,\mathit{endid},id_1)$ and $\sigma(L,\mathit{endid},id_2)$ share an event.

\medskip
We now discuss how to materialize such sequential information from a multi-entity event log $L$ for \emph{all} entity types $\mathit{ET}$. First, we canonically derive a set of sequential event logs of $L$, one per entity type $\mathit{et} \in \mathit{ET}$. In Sect.~\ref{sec:logs:multi_spo} we discuss an alternative view based on partial orders.

Note that the functions $\mathit{et}(L)$, $\mathit{corr}(L,\mathit{et},id)$, $\sigma(L,\mathit{et},id)$ of Def.~\ref{def:trace} are well-defined over multi-entity event logs.
\begin{definition}[Sequential views on multi-entity event log]\label{def:event_log_multiid:seq_view}
Let $L = (E,\mathit{AN},\mathit{ET},\#)$ be a multi-entity event log.

A \emph{sequential event log} of $L$ for entity type $\mathit{et} \in \mathit{ET}$ is a set $\sigma(L,\mathit{et})$ containing exactly one trace $\sigma \in \sigma(L,\mathit{et},id)$ for each case $id \in \mathit{et}(L)$ of $\mathit{et}$ (see Def.~\ref{def:trace}).

A \emph{sequential view} on $\,L\,$ is a family $\,\langle \sigma(L,\mathit{et}) \rangle_{\mathit{et} \in \mathit{ET}}\,$ of sequential event logs\,--\,one per entity type in~$L$.
\end{definition}
As for sequential event logs, if $L$ is monotone, then the sequential event log $\sigma(L,\mathit{et})$ of an entity is unique, and the sequential view on $L$ is unique.

The sequential view on the monotone multi-entity event log in Tab.~\ref{tab:event_table_multiid} has three sequential event logs $\sigma(L,pid)$, $\sigma(L,qid)$, $\sigma(L,rid)$. Thereby the $pid$-log $\sigma(L,pid)$ is the same as for the single-entity event log. It describes the behavior along the classical case identifier, i.e., one trace per bag in the system.

Log $\sigma(L,rid)$ has cases $m3,m4,d1,d2$ and, among others, traces $\sigma(L,rid,m4) = \sigma_{rid}^{m4} = \langle e_3,e_4,e_5,e_6 \rangle$ and $\sigma_{rid}^{d1}  = \langle e_7,e_8,e_9,e_{10} \rangle$. These traces describe the order in which each machines was used. Note that $e_3,e_4,e_7,e_8 \in \mathit{corr}(L,pid,50)$ while $e_5,e_6,e_9,e_{10} \in \mathit{corr}(L,pid,51)$. That is, traces $\sigma_{rid}^{m4}$ and $\sigma_{rid}^{d1} $ for $rid$ contain events from different bags, i.e., the $rid$-traces go ``across'' multiple different $pid$-traces.

Log $\sigma(L,qid)$ has cases $c3:m3,m3:m4,c4:m4,m4:d1,d1:s1,d1:d2,d2:s2$ and, among others, trace $\sigma_{qid}^{m4:d1} = \langle e_4,e_6,e_7,e_9 \rangle$ while $e_4,e_7 \in \mathit{corr}(L,pid,50)$ and $e_6,e_9 \in \mathit{corr}(L,pid,51)$. These traces describe the order in which different $pid$-cases entered and left the queues.

Note that per sequential event log, each event occurs in only one trace, but the same event can be part of multiple different event logs (for different entity types), e.g., $e_4$ occurs in $\sigma_{pid}^{50}$, $\sigma_{rid}^{m4}$, and $\sigma_{qid}^{m4:d1}$. In this way, the $rid$- and $qid$-traces describe how different $pid$-traces are synchronized via shared machines ($rid$) and conveyor belts ($qid$). However, that synchronization of multiple traces is implicit in the sequential view. We therefore propose a partially-ordered view on a multi-entity event log next.

\subsection{Partially ordered view on event logs over multiple entities}
\label{sec:logs:multi_spo}

In Sect.~\ref{sec:logs:multi_sequential}, we used $\#_{et}(.)$ to derive sequences of events related to the same entity ordered by $\#_{time}(.)$. We next encode the same information in an ordering relation over events, which is a strict partial order due to the monotonicity of the $\#_{time}(.)$ values. We thereby start by first ordering events $e_1 < e_2$ only if they are related to the same entity. The transitive closure then naturally extends this ordering across multiple different entities.

%
\begin{definition}[Partial-order view, system-level run]\label{def:event_log_multiid:po_view}\label{def:precedes}\label{def:system-level_run}
Let $L = (E,\mathit{AN},\mathit{ET},\#)$ be a monotone multi-entity event log.

\medskip
Let $\mathit{et} \in \mathit{ET}$ and $\mathit{id} \in \mathit{et}(L)$. Event $e_1 \in E$ \emph{precedes} event $e_2 \in E$ in entity $id$ of type $et$, written $e_1 <_{et}^{id} e_2$ iff
\begin{enumerate}
  \item $\perp \neq \#_{time}(e_1) < \#_{time}(e_2) \neq \perp$ (the time stamp of $e_1$ is before the time stamp of $e_2$), and
  \item $\#_{et}(e_1) = \#_{et}(e_2) = id$ (both events are related to the same entity $id$).
\end{enumerate}
$e_1$ \emph{directly precedes} $e_2$ in entity $id$ of type $et$, written $e_1 \lessdot_{et}^{id} e_2$, iff there exists no $e' \in E$ with $e_1 <_{et}^{id} e' <_{et}^{id} < e_2$.

\medskip
This ordering lifts to entity types and entire $L$:
\begin{itemize}
\item $e_1$ \emph{directly precedes} $e_2$ in entity type $et$, written $e_1 \lessdot_{et} e_2$, iff $e_1 \lessdot_{et}^{id} e_2$ for some $id \in et(L)$; and
\item $e_1$ \emph{directly precedes} $e_2$, written $e_1 \lessdot e_2$, iff $e_1 \lessdot_{et} e_2$ for some $et \in \mathit{ET}$.
\item The transitive closures $(\lessdot_{et})^+ = <_{et}$ and $(\lessdot)^+ = <$ define \emph{(indirectly) precedes} per entity type and for all events in $L$, respectively.
\end{itemize}
The partial-order view on $L$ or \emph{system-level run} of $L$ (induced by $\tau$) is $\pi = (E, <,\mathit{AN},\mathit{ET},\#)$.
\end{definition}
Figure~\ref{fig:example}(b) visualizes the directly precedes relations $\lessdot_{pid}, \lessdot_{rid}, \lessdot_{qid}$ induced by $\#_{time}$ for the multi-entity event log in Tab.~\ref{tab:event_table_multiid}. This behavioral model shows that events of different process cases ($pid=50$ and $pid=51$) are independent under the classical control-flow perspective $<_{pid}$, e.g., $e4 \not<_{pid} e5 \not<_{pid} e7$ (see Def.~\ref{def:precedes}), but mutually depend on each other under $<_{rid}$ and $<_{qid}$, e.g., $e4 <_{rid} e5 <_{rid} e6$ and $e6 <_{qid} e7$.

Note that events without defined time stamp are unordered to all other events, i.e., they can occur at any time. We will exploit this when inferring missing time stamps.

The order $<$ is indeed a strict partial order.
\begin{lemma}
Let $L = (E,\mathit{AN},\mathit{ET},\#)$ be a monotone multi-entity event log. Let $\pi = (E, \mathord{<},\mathit{AN},\mathit{ET},\#)$ be the system-level run of $L$. Then $(E,<)$ is a strict partial order.
\end{lemma}
\begin{proof}
We have to show that $<$ is transitive and irreflexive. $< = (\lessdot)^+$ is transitive by construction in Def.~\ref{def:event_log_multiid:po_view}. Regarding irreflexivity: $e_1 \lessdot e_2$ holds only if $\#_{\tau}(e_1) < \#_{\tau}(e_2)$. As $L$ is monotone, either $\#_{\tau}(e_1) < \#_{\tau}(e_2)$ or $\#_{\tau}(e_2) < \#_{\tau}(e_1)$ holds (Def.~\ref{def:time-monotone}) but not both, hence $<$ is irreflexive.
\end{proof}

\subsection{Relation between sequential and partially-ordered view}
\label{sec:logs:multi_sequential_spo}

To better define and solve the problem, we now establish a more explicit relation between the system-level run $\pi$ of $L$ and the traces in the sequential view of $L$.

\medskip
Given a system-level run $\pi = (E, <,\mathit{AN},\mathit{ET},\#)$ we write $\pi_{et}$ for the projection of $\pi$ onto entity type $et \in \mathit{ET}$ where $\pi_{et} = (E_{et},<_{et},\mathit{AN},\{et\},\#)$ contains only the events $E|_{et} = \{ e \in E \mid \#_{et}(e) \neq \perp \}$ related to $\mathit{et}$. Relation $<_{et}$ is already well-defined wrt. $E|_{et}$. We call $\pi_{et}$ the entity-type level run of $L$ for entity type $et$.

Correspondingly, given an identifier $id \in et(L)$, the projection $\pi_{et}^{id} = (E_{et}^{id},<_{et}^{id},\mathit{AN},\{et\},\#)$ contains only the events $E|_{et}^{id} = \mathit{corr}(L,et=id)$ of $id$.  We call $\pi_{et}^{id}$ the entity-level run of $L$ of entity $id$ of type $et$.

From the system-level run in Figure~\ref{fig:example}(b) we can obtain the entity-level runs $\pi_{pid}^{50}$ and $\pi_{pid}^{51}$ from the perspective of the process, $\pi_{rid}^{m3}, \pi_{rid}^{m4}, \pi_{rid}^{d1}, \pi_{rid}^{d2}$ from the perspective of the resources, and $\pi_{qid}^{c3:m3},\pi_{qid}^{m3:m4},\pi_{qid}^{c4:m4},\pi_{qid}^{m4:d1},\pi_{qid}^{d1:d2}, \pi_{qid}^{d1:s1},\pi_{qid}^{d2:s2}$ from the perspective of the conveyor belts (or queues).

Each entity-level run $\pi_{et}^{id}$ corresponds to a sequential trace $\sigma_{et}^{id}$ in the sequential view of $L$ because either view derives the direct precedence/succession of events from the same principles.
\begin{lemma}
Let $L = (E,\mathit{AN},\mathit{ET},\#)$ be a monotone multi-entity event log. Let $\pi = (E, \mathord{<},\mathit{AN},\mathit{ET},\#)$ be the system-level run of $L$.

For all $e_1,e_2 \in E$ holds: $e_1 \lessdot e_2$ iff there exists $et \in \mathit{ET} $ and $id \in et(L)$ so that $\langle \ldots, e_1,e_2, \ldots \rangle = \sigma(L,et=id)$ is a trace in the sequential view $\langle \sigma(L,et) \rangle_{et \in \mathit{ET}}$ of $L$.
\end{lemma}
\begin{proof}
If $e_1 \lessdot e_2$ then by Def.~\ref{def:event_log_multiid:po_view}, $e_1 \lessdot_{et}^{id} e_2$ for some $et \in \mathit{ET} $ and $id \in et(L)$. Thus, $\#_{et}(e_1) = \#_{et}(e_2)$ and $\#_{time}(e_1) < \#_{time}(e_2)$ (by Def.~\ref{def:event_log_multiid:po_view} and $L$ being monotone). By Def.~\ref{def:trace}, $e_1,e_2 \in \mathit{corr}(L,et=id)$ (correlated into the same case $id$ for $et$). Further, because there is no $e' \in \mathit{corr}(L,et=id)$ with $\#_{time}(e_1) < \#_{time}(e') < \#_{time}(e_2)$ (definition of $\lessdot$ in Def.~\ref{def:event_log_multiid:po_view}), $e_1$ and $e_2$ are ordered next to each other in the sequential trace $\langle \ldots, e_1,e_2, \ldots \rangle = \sigma(L,et=id)$. The converse holds by the same arguments.
\end{proof}
\begin{corollary}\label{cor:sequence_po_equivalence}
Let $L = (E,\mathit{AN},\mathit{ET},\#)$ be a monotone multi-entity event log. Let $\pi = (E, \mathord{<},\mathit{AN},\mathit{ET},\#)$ be the system-level run of $L$. For any $\pi_{et}^{id}$ for $et \in \mathit{ET}, id \in et(L)$ holds $e_1 \lessdot_{et}^{id} e_2$ iff $\langle \ldots, e_1,e_2, \ldots \rangle = \sigma(L,et=id)$ and $e_i <_{et}^{id} e_j$ iff $\langle \ldots, e_i, \ldots, e_j, \ldots \rangle = \sigma(L,et=id)$.
\end{corollary}
The above relation may not seem profound and be more a technical exercise. However, we benefit in the next sections from being able to change perspectives at will and study (and operate on) behavior as a classical sequence (and use sequence reasoning for a single entity) as well as a partial order (and use partial order reasoning across different entities).

\medskip
For instance, the directly precedes relations $\lessdot_{pid}, \lessdot_{rid}, \lessdot_{qid}$ visualized in Figure~\ref{fig:example}(b) directly define the sequences of events we find in $\sigma(L,pid)$, $\sigma(L,rid)$, and $\sigma(L,qid)$.


\section{Modeling multi-entity behavior with queueing and time}
\label{sec:pqr}

We introduced system-level runs and multi-entity event logs in Sect.~\ref{sec:logs}. We now want to formally state the problem of inferring missing time-stamps from events logs which only recorded partial information about the system. Thereby the gap between ``partial'' and ``complete'' information depends on domain knowledge. We therefore first introduce a model for describing such domain knowledge about system-level behavior and state the formal problem afterwards.

\medskip
Our model uses synchronous proclets~\cite{DFahlandMultiDim} to describe a system as a composition of multiple smaller components (each called a proclet) that synchronize dynamically on transition occurrences. Which transitions synchronize is specified in channels. The original definition~\cite{DFahlandMultiDim} is based on Petri nets with identifiers. To model queueing and time, we extend the synchronous proclet model with concepts of coloured Petri nets (CPN).

We first recall some basic syntax of colored Petri nets in Sect.~\ref{sec:cpn_syntax} and then formulate a \emph{replay semantics} to replay an event log over a CPN in Sect.~\ref{sec:cpn_replay}. This replay semantics allows us to define conformance checking problems between a CPN and a multi-entity event log.

We then lift this definition of CPN replay semantics to CPN proclet systems where multiple proclets (each defined by a CPN) synchronize on joint transition firings in Sect.~\ref{sec:pqr:cpn_proclets}; the syntax defined there extend the basic synchronous proclet model~\cite{DFahlandMultiDim} with concepts for data and time.

We then consider a specific class of CPN proclet systems which describe a single process with shared resources and queues. We call such a proclet system a \emph{PQR system}. We introduce PQR systems in Sect.~\ref{sec:pqr:systems} and then formally state our research problem of repairing incomplete event logs in Sect.~\ref{sec:pqr:problem_statement}.

\subsection{Background on coloured Petri nets}
\label{sec:cpn_syntax}

We here only recall the CPN notation and semantic concepts also needed in the remainder of this paper and do not introduce the entire formal model CPN; refer to~\cite{DBLP:journals/cacm/JensenK15} for an introduction and further details.

\medskip
A \emph{labeled coloured Petri net} (CPN) $N = (P,T,F,\Sigma,\ell,\mathit{Var},\mathit{Types},\mathit{colSet},m_0,\mathit{arcExp},\mathit{arcTime})$ defines
\begin{itemize}
\item a \emph{skeleton Petri net} $(P,T,F)$ of places $P$, transitions $T$, and arcs $F$ as usual; we write $\PNpre{t}$ and $\PNpost{t}$ for the pre- and post-places of transition $t$ and $\PNpre{p}$ and $\PNpost{p}$ for the pre- and post-transitions of place $p$;
\item a \emph{labelling function} $\ell : T \to \Sigma$ assigning each transition a name $a \in \Sigma$ from an alphabet $\Sigma$;
\item a set of \emph{variable} names $\mathit{Var}$ and a set of data $\mathit{Types}$;
\item \emph{color sets} (i.e., data types) $\mathit{colSet} : P \cup \mathit{Var} \to \mathit{Types}$ specifying for each place and variable which values it can hold;
\item an \emph{initial marking} $m_0 : P \to 2^{\mathit{Values} \times \mathds{R}}$ assigning to each place a multiset of value-time pairs $(v,time)$ so that $v \in \mathit{colSet}(p), p \in P$;
\item \emph{arc expressions} $\mathit{arcExp} : F \to \mathit{Exp}$ defining for each arc an expression over $\mathit{Var}$ and various operators, specifying which values to consume/produce; and
\item a \emph{time annotation} $\mathit{arcTime} : F \to \mathds{R}$ defining for an output arc $(t,p)$ how much time $\mathit{arcTime}(t,p)$ has to pass until a produced token becomes available.
\end{itemize}
Figure~\ref{fig:bhs} shows multiple labeled coloured CPNs:

\smallskip
In \emph{Process} (shown in red) each arc is annotated with the variable $\mathit{pid}$ and each place $p$ has the colorset $\emph{colSet}(p) = \mathit{Pid}$ describing identifiers for different bags (cases) of the process. Thus, bags are described by their identifiers and transitions consume and produce these identifiers, thereby moving the bag forward through the process. Note that some transitions have no pre-places, e.g., $c1_s$. The $pid$ value that these transitions produce can be chosen freely; our replay semantics in Sect.~\ref{sec:cpn_replay} will determine the $pid$ value based on an event log.

\smallskip
In \emph{CheckIn-c1}, places \emph{idle} and \emph{busy} have colorset $\mathit{Rid}$. Place \emph{idle} carries a token $(c1,0)$ meaning value $c1$ is available from time $0$ onwards, i.e., resource $c1$ is idle at this time.
All arcs carry variable $rid$ as arc expression. Arc $(c1_s,busy)$ carries time annotation $\mathit{arcTime}(c1_s,busy) = tsr_{c1} > 0$ which specifies a delay of $tsr_{c1}$ time units. When check-in starts ($c1_s$ fires) the resource $c1$ moves from \emph{idle} to \emph{busy} and only becomes available after $tsr_{c1}$ time units. Then check-in can complete and $c1$ moves from \emph{busy} to \emph{idle} and only becomes available after $twr_{c1}$ time units.

\smallskip
In \emph{Queue-c1:m2}, place \emph{queue} has a colorset of a pair $(qid,q)$ where $qid \in \mathit{Qid}$ is a queue identifier and $q \in \mathit{Pid^*}$ is a list of bag identifiers. The initial token on \emph{queue} is $(c1:m2,\langle \rangle)$ (i.e., the empty queue). Transition $c1_c$ places a new bag ($pid$) on the start of conveyor belt by adding it to the end of the current queue, transition $m1_s$ removes a bag from the end of conveyor belt by removing its head from the queue. When $c1_c$ fires, the current queue $(qid,q), q=\langle pid_1,\ldots,pid_n\rangle$ is consumed from place \emph{queue} and a bag identifier $pid$ is appended to $q$, producing $(qid, q\hat{\ }\hat{\ } \langle pid \rangle) = \langle pid_1,\ldots,pid_n,pid\rangle$ onto place \emph{queue}. At the same time, token $pid$ is produced on $p_3$ with a delay of $twq_{c1:m2}$, i.e., $pid$ only becomes available after $twq_{c1:m2}$ time units. This allows to model that conveyor belt movement takes time. When this delay for $pid$ has passed and $pid$ is at the head of the queue, i.e., token $(qid,q')$ with $q' = pid::\langle pid_1,\ldots,pid_n\rangle = \langle pid, pid_1,\ldots,pid_n\rangle$ is on place \emph{queue}, then $m1_s$ can fire. If $m1_s$ fires it consumes $(qid,q')$ from \emph{queue} and $pid$ from $p_3$, and produces $(qid,q)$ with $q = \langle pid_1,\ldots,pid_n\rangle$ on \emph{queue}, thereby removing $pid$ from the queue.\vspace*{-1mm}

\subsection{An event log replay semantics for colored Petri nets}
\label{sec:cpn_replay}

We briefly recall the semantic concepts of CPNs. Let $N = (P,T,F,\Sigma,\ell,\mathit{Var},\mathit{Types},\mathit{colSet},m_0,\mathit{arcExp},\mathit{arcTime})$  be a CPN in the following.

\smallskip
A \emph{timed} marking $m$ assigns each place $p$ a multiset $m(p)$ of timed tokens $(val,time')$ where $val$ is the value on $p$ and $time'$ is the time after which $v$ can be consumed. A \emph{state} $s = (m,time)$ of a $N$ has a timed marking $m$ and a time-stamp $time$. The time-stamp $time$ is the global system time reached. The initial state of $N$ is $(m_0,0)$.

A \emph{binding} $\beta : \mathit{Var} \to \mathit{Values}$ binds each variable $v$ to some value $\beta(var)$. An arc expression $exp$ can be evaluated under a binding $\beta$ by replacing each variable $var$ in $exp$ with $\beta(var)$ and computing the result, we denote the result by $exp[\beta]$.

A transition $t$ is \emph{enabled} in a state $s = (m,time)$ for binding $\beta$ iff for each input arc $(p,t) \in F$ holds there exists $(val,time') \in m(p)$ so that $\mathit{arcExp}(p,t)[\beta] = val$ and $time' \leq time$, i.e., evaluating the arc expression $\mathit{arcExp}(p,t)$ yields a value $val$ which is already available at the current time.

\medskip
If $t$ is enabled in $s$ for binding $\beta$, then $t$ can fire resulting in the \emph{transition-step} $s = (m,time) \xrightarrow{t,\beta} (m',time) = s'$ where
\begin{enumerate}
\itemsep=0.9pt
\item for each pre-place $p$ of $t$, remove from $m(p)$ a timed token $(\mathit{arcExp}(p,t)[\beta],time')$ with $time' \leq time$, resulting in an intermediate marking $m''$, and
\item for each post-place $p$ of $t$, add to $m''(p)$ the token $(\mathit{arcExp}(t,p),time + \mathit{arcTime}(t,p))$.
\end{enumerate}
Further $N$ can make a \emph{time-step} $s = (m,time) \xrightarrow{delay} (m,time+delay) = s', delay \geq 0$. In the original CPN semantics~\cite{DBLP:journals/cacm/JensenK15} time has to advance at most until the next transition becomes enabled. We drop this requirement and allow arbitrary time-progress to facilitate replaying event logs with their own time-stamps.

We can now define the semantics of replaying a multi-entity event log $L = (E,\mathit{AN},\mathit{ET},\#)$ over a CPN. The idea is that the activity name $\#_{act}(e)$ of an event $e$ specifies the label $\ell(t)$ of the transition $t$ that shall be fired, $\#_{time}(e)$ specifies the global time when $t$ fires. We treat the attributes $\mathit{AN} \setminus \{ time, act\}$ as variables and the attribute-value function $\#_a(e) = v$ defines the binding $\beta(a) = v$ for which $t$ shall fire. For this definition we will ignore the entity types $\mathit{ET}$ as the definition is general to any CPN.
\begin{definition}[CPN Replay Semantics]
Let $N = (P,T,F,\Sigma,\ell,\mathit{Var},\mathit{Types},\mathit{colSet},m_0,\mathit{arcExp},\mathit{arcTime})$  be a CPN with initial state $s_0 = (m_0,0)$. Let $L = (E,\mathit{AN},\mathit{ET},\#)$ be a time-complete multi-entity event log, i.e., each event has a time-stamp.

\smallskip
Let $\sigma = \langle e_1,\ldots,e_n\rangle$ be a sequence of all events in $E$ such that $\#_{time}(e_i) \leq \#_{time}(e_{i+1})$ for $1 \leq i < n$.
Each event $e_i$ defines the \emph{binding} $\beta_i$ with $\beta_i(a) = \#_a(e_i)$ for all $a \in \mathit{AN} \setminus \{ time, act\}$.

An event sequence $\sigma$ of $L$ can be \emph{replayed} on $N$ iff for $1 \leq i \leq n$ exist delay $0 \leq d_i \in \mathds{R}$ and transition $t_i \in T$ with $\ell(t) = \#_{act}(e_i)$ so that $s_{i-1} = (m_{i-1},time_{i-1}) \xrightarrow{d_i} (m_{i-1},\#_{time}(e_i)) \xrightarrow{t_i,\beta_i} (m_{i},\#_{time}(e_i)) = s_i$ are time- and transition-steps in $N$ that advance to $\#_{time}(e_i)$ and fire $\#_{act}(e_i)$. 

By $t(e_i,\sigma)$ we denote for each event $e_i$ the transition $t_i \in T$ which performed the transition-step described by $e_i$ when replaying $\sigma$.
\end{definition}
Note that this definition requires that the transition $t = t(e_i,\sigma)$ is enabled at time $\#_{time}(e_i)$ (or was already enabled earlier).

\medskip
For example, the sequential trace $\sigma_{qid}^{c3:m3} = \langle e_0, e_1 \rangle$ of the event log in Table~\ref{tab:event_table_multiid} is replayed on the CPN \emph{Queue-c3:m3} (identical to CPN \emph{Queue-c3:m3} in Fig.~\ref{fig:bhs}) as follows:
\begin{itemize}
\itemsep=0.9pt
\item the initial marking is $m_0(queue) = [((c3:m3,\langle \rangle),0)], m_0(p3) = []$;
\item $e_0$ yields steps $(m_0,0) \xrightarrow{9\mathord{:}00\mathord{:}15} (m_0, 9\mathord{:}00\mathord{:}15) \xrightarrow{c3_c,\beta_1} (m_1, 9\mathord{:}00\mathord{:}15)$ with binding $\beta_1(qid) = c3:m3, \beta_2(pid) = 50$ and resulting marking $m_1(queue) = [((c3:m3,\langle 50 \rangle), 9\mathord{:}00\mathord{:}15)], m_1(p3) = [ (50,9\mathord{:}00\mathord{:}15+twq_{c3:m3}) ]$;
\item $e_1$ yields steps $(m_1, 9\mathord{:}00\mathord{:}15) \xrightarrow{0\mathord{:}00\mathord{:}15} (m_2, 9\mathord{:}00\mathord{:}30) \xrightarrow{m3_s,\beta_2} (m_2, 9:00:30)$ with binding $\beta_2(qid) = c3:m3, \beta_2(pid) = 50$ and resulting marking $m_2(queue) = [((c3:m3,\langle \rangle), 9\mathord{:}00\mathord{:}30)]$ and $m_2(p3) = [ ]$.
\end{itemize}

\subsection{Synchronous proclet systems with data and time}
\label{sec:pqr:cpn_proclets}

We can now define our model of a synchronous proclet system with data and time.

\begin{definition}[CPN Proclet]
A \emph{CPN proclet} $\mathit{Proc} = (N,et)$ is a CPN $N = (P,T,F,\Sigma,\ell,\mathit{Var},\mathit{Types},\mathit{colSet},m_0,\mathit{arcExp},\mathit{arcTime})$  so that
\begin{itemize}
\item $et \in Var$ is a designated \emph{entity type variable} that can be bound to entity identifier values $\mathit{colSet}(et) = \mathit{IDval}$,
\item for each transition $t$ with a pre-place, there exists an arc $(p,t)$ with $\mathit{arcExp}(p,t)$ has the form $et$ or $(et,exp)$ where $exp$ is some expression, and
\item for each transition $t$ with a post-place, there exists an arc $(t,p)$ with $\mathit{arcExp}(t,p)$ has the form $et$ or $(et,exp)$ where $exp$ is some expression.
\end{itemize}
\end{definition}
All CPNs in Fig.~\ref{fig:bhs} are CPN-proclets, e.g., \emph{Queue-c1:m2} has $\mathit{et} = qid$ and $c1_c$ and $m1_s$ both have incoming and outgoing arcs of the form $(qid,exp)$; the arcs to/from $p_3$ have a different form.

This structure ensures that each transition in a CPN-proclet $\mathit{Proc} = (N,et)$ occurs in relation to a specific entity instance identified by variable $et$, e.g., a specific bag, resource, or queue. A \emph{proclet system} synchronizes multiple proclets via channels.
\begin{definition}[CPN Proclet System]\label{def:cpn_proclet_system}
A \emph{CPN proclet system} $S = (\{\mathit{Proc}_1,\ldots,\mathit{Proc}_k\},C)$ is a set $\{\mathit{Proc}_1,\ldots,\mathit{Proc}_k\}$ of proclets with disjoint sets of transitions and places, and a set of \emph{synchronization channels} $C \subseteq 2^T$ being sets of transitions so for each channel $\{t_1,\ldots,t_r\}$ holds $\ell(t_i) = \ell(t_j)$.
\end{definition}
In Fig.~\ref{fig:bhs} shows a CPN proclet system where the channels are indicated by dashed edges, e.g., all transitions labeled $c1_c$ in \emph{CheckIn-c1}, \emph{Process}, and \emph{Queue-c1:m2} form a channel.

The intuition is that transitions connected via a channel $(t_i,t_j)$ can only fire together, i.e., proclets $\mathit{Proc}_i$ and $\mathit{Proc_j}$ must each be in a marking where $t_i$ and $t_j$ are enabled for the same binding (i.e., variables occurring in both $t_i$ and $t_j$ must be bound to the same values). While the original proclet semantics~\cite{DFahlandMultiDim} is an operational semantics, we now provide a replay semantics over a CPN proclet system.

We replay a multi-entity event log $L = (E,\mathit{AN},\mathit{ET},\#)$ over a CPN proclet system $S = (\{\mathit{Proc}_1,\ldots,\mathit{Proc}_k\},C)$ by decomposing $L$ into its sequential event logs $L_{et},et\in\mathit{ET}$ and replaying each sequential event log over the corresponding proclet in $S$. As multiple proclets may use the same entity type $et$ (e.g., \emph{CheckIn-c1} and \emph{MergeUnit-m2} both use $rid$), we need to specify which case in $L_{et}$ belongs to which proclet $\mathit{Proc}_i$ in $S$.
\begin{definition}[Replaying a Log over a CPN Proclet System]\label{def:cpn_proclet_system:replay_semantics}
Let $L = (E,\mathit{AN},\mathit{ET},\#)$ be a multi-entity event log (Def.~\ref{def:event_log_multiid}) that is time-complete (Def.~\ref{def:event_log_time-incomplete}). Let $S = (\{\mathit{Proc}_1,\ldots,\mathit{Proc}_k\},C)$ be a proclet system.

Let $f : \mathit{Val} \to \{1,\ldots,k\}$ be a mapping so that for each $\mathit{et} \in \mathit{ET}$ and each $id \in et(L)$, $f(id) = i$ maps to a proclet $\mathit{Proc}_i$ with $et_i = et$.

\eject
The entire log $L$ can be \emph{replayed} over $S$ (for a given mapping $f$) iff the following conditions hold:
\begin{enumerate}
\item For each $et \in \mathit{ET}$ and sequential event log $\sigma(L,et)$ of $L$ (Def.~\ref{def:event_log_multiid:seq_view}) holds: each trace $\sigma_{et}^id \in \sigma(L,et)$ can be replayed on the CPN $N_i$ of proclet $\mathit{Proc}_i$ with $i = f(id)$.
\item for each event $e \in E$ and all traces $\sigma_{et_1}^{id_1},\ldots,\sigma_{et_k}^{id_k}$ that contain $e$, the set $c = \{ t(e,\sigma_{et_1}^{id_1}),\ldots,$ $t(e,\sigma_{et_k}^{id_k})\}$ of transitions replayed in the different traces is either singleton $c = \{ t \}$ or a channel $c \in C$ of $S$.
\end{enumerate}
\end{definition}
Figure~\ref{fig:example} shows how the multi-entity event log $L$ of Tab.~\ref{tab:event_table_multiid} can be replayed over the proclet system of Fig.~\ref{fig:bhs}. Each dashed rectangle in Figure~\ref{fig:example}(a) abstractly illustrates how one sequential trace in $L$ is replayed over one of the proclets in Fig.~\ref{fig:bhs}, the circles indicate parts of the markings reached after replaying each event, e.g., replaying $e0$ in proclet \emph{Queue-c3:m3} yields a token $(c3:m3,[50])$ on place $queue$, etc. The dashed lines indicate how the channel constraints are satisfied by this replay. For instance, event $e_4$ is replayed by transition $m4_c$ in proclet $\mathit{Queue-c4:m4}$ (trace $\sigma_{qid}^{c4:m4}$), by $m4_c$ in $\mathit{Process}$ (trace $\sigma_{pid}^{50}$), and by $m4_c$ in $\mathit{MergeUnit-m4}$ (trace $\sigma_{rid}^{m4}$); all three $m4_c$ transitions form a channel in the proclet system in Fig.~\ref{fig:bhs}.

\subsection{PQR systems}
\label{sec:pqr:systems}

In the following, we only study a specific sub-class of CPN proclet systems which describes processes with shared resources and queues, which we call PQR-systems. Each PQR system is composed of one process proclet and multiple resource and multiple queue proclets in a specific way. Figure~\ref{fig:bhs} shows an example of a PQR system. We define each proclet type first and then the entire composition.

\medskip
Intuitively, a process proclet describes a sequential process where each process step has designated \emph{start} and \emph{complete} transitions, i.e., each step is non-atomic and start and complete are separately observable. Moreover, the process proclet allows creating arbitrarily many fresh process instances through source transitions without pre-places; cases that complete are consumed by sink transitions without post-places. This is different from the concept of workflow nets~\cite{Aalst2004WorkflowMD} which model only the evolution of a single case and abstract from case creation and deletion.

\begin{definition}[Process proclet]\label{def:p-proclet}
A \emph{Process-proclet} (or P-proclet) is a CPN proclet $(N,pid), N = (P,T,F,\Sigma,\ell,\mathit{Var},\mathit{Types},\mathit{colSet},m_0,\mathit{arcExp},\mathit{arcTime})$ where the following properties hold:
\begin{enumerate}
  \item $P = P_{activity} \uplus P_{handover}$ (places either describe that an activity is being executed or that a case being handed over to the next activity);
  \item $T = T_{start} \uplus T_{complete}$ (transitions either describe that an activity is being started or being completed)
  \item $N$ is a state-machine, i.e., $|\PNpre{t}| \leq 1$ and $|\PNpost{t}| \leq 1$ and all nodes are connected,
  \item $N$ is transition-bordered, i.e,. $|\PNpre{p}| \geq 1$ and $|\PNpost{p}| \geq 1$ and the sets $T_{source} = \{ t \in T \mid \PNpre{t} = \emptyset \} \neq \emptyset$ and $T_{sink} = \{ t \in T \mid \PNpost{t} = \emptyset \} \neq \emptyset$ of \emph{source and sink transitions} are non-empty.
  \item Each activity place $p \in P_{activity}$ is only entered via start transitions and only left with complete transitions, i.e., $\PNpre{p} \subseteq T_{start}$ and $\PNpost{p} \subseteq T_{complete}$.
  \item Each handover place $p \in P_{handover}$ is only entered via exactly one complete transitions (of the preceding activity) and left only via exactly one start transition (of the succeeding activity), i.e., $\PNpre{p} = \{ t \} \subseteq T_{complete}$ and $\PNpost{p} = \{ t \} \subseteq T_{start}$.
  \item All arcs carry the entity identifier $pid$: $\mathit{arcExp}(x,y) = pid \in \mathit{Var}$ for all $(x,y) \in F$.
  \item No place carries an initial token: $m_0(p) = []$ for all $p \in P$.
\end{enumerate}
\end{definition}
The proclet \emph{Process} in Figure~\ref{fig:bhs} is a P-proclet.

\medskip
In this paper, a resource proclet defines the most basic life-cycle of a shared resource: the resource is initially \emph{idle} (available to do work), then starts an activity making the resource \emph{busy}. There is a \emph{minimum service time} $tsr$ the resource is busy before the task completes. After completing the task, the resource is \emph{idle} again but requires a \emph{minimum waiting time} $twr$ before being able to work again. Figure~\ref{fig:bhs} shows several resource proclets which we formally capture in the following definition.

\begin{definition}[Resource proclet]
A \emph{Resource-proclet} (or R-proclet) is a CPN proclet $(N,rid), N = (P,T,F,\Sigma,\ell,\mathit{Var},\mathit{Types},\mathit{colSet},m_0,\mathit{arcExp},\mathit{arcTime})$  for a resource with minimum service time $tsr$ and minimum waiting time $twr$ when the following properties hold:
\begin{enumerate}
  \item $P = \{p_{idle},p_{busy}\}$;
  \item $T = T_{start} \uplus T_{complete}$ (transitions either describe that an activity is being started or being completed);
  \item $\PNpost{p}_{idle} = T_{start} = \PNpre{p}_{busy}$ (resources go from idle to busy via start transitions);
  \item $\PNpost{p}_{busy} = T_{complete} = \PNpre{p}_{idle}$ (resources go from busy to idle via start transitions);
  \item $\mathit{arcExp}(x,y) = rid$ for all $(x,y) \in F$
  \item $\mathit{arcTime}(t_{start},{p}_{busy}) = tsr$ for all $t_{start} \in T_{start}$ and $\mathit{arcTime}(t_{complete},{p}_{idle}) = twr$ for all $t_{complete} \in T_{complete}$
  \item $m_0(p_{idle}) = [rid]$ and $m_0(p_{busy}) = []$ (the resource is idle initially)
\end{enumerate}
\end{definition}
An R-proclet has multiple start and complete transitions to mirror that an activity in a P-proclet has multiple start and complete transitions. This will simplify the composition of proclets later on. For example \emph{Merging-Unit-m2} in Fig.~\ref{fig:bhs} has two start and one complete transitions while \emph{Diverting-Unit-d1} has one start and two complete transitions.

\medskip
In this paper, a queue proclet defines the most basic operation of a queue: it ensures that items leave the queue in the order in which they entered the queue; moreover, we specify that traversing the queue requires a minimal waiting time $twq$.
\begin{definition}[Queue proclet]
A \emph{Queue-proclet} (or Q-proclet) is a CPN proclet $(N,qid), N = (P,T,F,\Sigma,\ell,\mathit{Var},\mathit{Types},\mathit{colSet},m_0,\mathit{arcExp},\mathit{arcTime})$  for a queue identified by entity identifier value $Q$ with minimum waiting time $twq$ when $N$ is an instance of the CPN template shown in Fig.~\ref{fig:q-proclet}.
\begin{figure}[h]
  \centering
  \includegraphics[width=.4\linewidth]{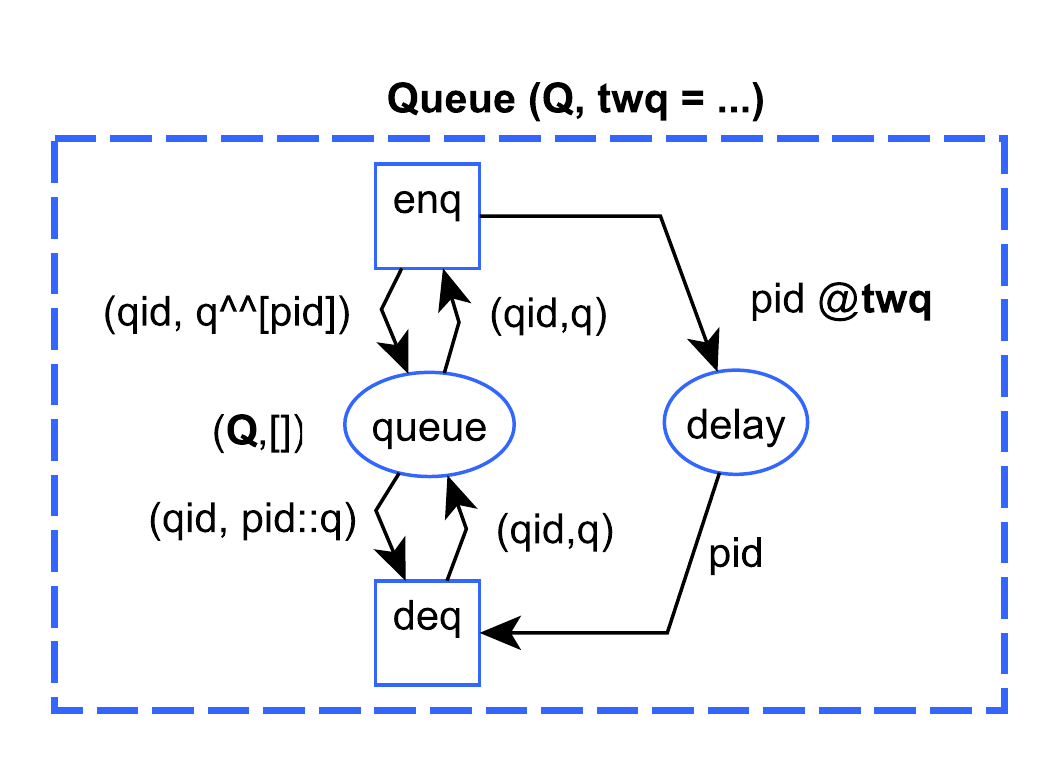}\vspace*{-2mm}
  \caption{CPN template for Queue proclet}\label{fig:q-proclet}
\end{figure}\vspace*{-1mm}
\end{definition}
For example proclet \emph{Queue-c1:m2} in Fig.~\ref{fig:bhs} is identified by entity identifier value \emph{c1:m2} and has minimum waiting time $\mathit{twq}_{c1:m2}$.

We can now formally define a PQR system as a CPN proclet system of one process proclet and multiple resource and queue proclets. A PQR system has specific synchronization constraints. Each activity in a process proclet (a place $p \in P_{activity}$ with corresponding start and complete transitions) synchronizes with one resource proclet which is responsible for executing this activity for any incoming case. Each handover between two activities in a process proclet (a place $p \in P_{handover}$) synchronizes with one queue moving cases from one activity to the next.

\begin{definition}[PQR system]\label{def:pqr-system}
A PQR system is a \emph{CPN proclet system} $S = (\{\mathit{Process}_0,R_1,\ldots,R_k,Q_{k+1},\ldots,Q_n\},C)$ where the following properties hold:
\begin{enumerate}
  \item $\mathit{Process}_0$ is a process proclet, $\mathit{R}_1,\ldots,\mathit{R}_k$ are resource proclets, $\mathit{Q}_{k+1},\ldots,\mathit{Q}_n$ are queue proclets.
  \item Each transition $t \in T_0 \cup \ldots \cup T_n$ is in exactly one channel $c \in C$, denoted by $c(t)$.
  \item For each activity place $p \in P_{activity}$ in $\mathit{Process}_0$ exists a resource proclet $R_i$ so that (1) for each $t_{start,0} \in \PNpre{p}$ (in $\mathit{Process}_0$) exists $t_{start,i} \in T_{start,i}$ (in $R_i$) with $c(t_{start,0}) = c(t_{start,i})$ and (2) for each $t_{complete,0} \in \PNpost{p}$ (in $\mathit{Process}_0$) exists $t_{complete,i} \in T_{complete,i}$ (in $R_i$) with $c(t_{complete,0}) = c(t_{complete,i})$.
  \item For each handover place $p \in P_{handover}$ in $\mathit{Process}_0$ exists a queue proclet $Q_i$ so that (1) transition $\{ t_{complete,0} \} = \PNpre{p}$ (in $\mathit{Process}_0$) synchronizes with $t_{enqueue,i} \in T_i$ (in $Q_i$) via $c(t_{complete,0}) = c(t_{enqueue,i})$ and (2) transition $\{ t_{start,0} \} = \PNpost{p}$ (in $\mathit{Process}_0$) synchronizes with $t_{dequeue,i} \in T_i$ (in $Q_i$) via $c(t_{start,0}) = c(t_{dequeue,i})$.
\end{enumerate}
\end{definition}
The proclet system in Fig.~\ref{fig:bhs} is a PQR system.

\eject
The above definition is rather declarative. To satisfy the above constraints there have to be enough R-proclets (one per activity place) of the correct shape (to match the start and complete transitions in the P-proclet), and enough Q-proclets (one per handover place). Moreover, the transition labels in all proclets have to match to form valid channels.

The definition enforces that in a PQR system each activity in a process is carried out by a shared resource of limited capacity (condition 3 in Def.~\ref{def:pqr-system}). Thus when multiple case arrive at the same activity, only for one of them the activity can be started while the others have to wait. Further, when an activity for a case is completed, the case enters a queue and can only reach the next activity when all other cases before it have reached that activity (condition 4 in Def.~\ref{def:pqr-system}).

Many real-life processes show more general use of shared resources and handover of cases than these very strict constraints. However, they are satisfied by material handling systems such as baggage handling systems. Generalizing the definition to other types of processes with shared resources is beyond the scope of this paper.

\subsection{Restoring partial event logs of a PQR system}
\label{sec:pqr:problem_statement}
\label{sec:problem_statement}

We can now formally state our research problem.

\medskip
Let $L$ be a multi-entity event log. And let $S$ be a PQR system defining proclets for a process (with case identifier $pid$), multiple resources (with entity identifier $rid$) and queues (with entity identifier $qid$); see Def.~\ref{def:event_log_multiid} and Def.~\ref{def:pqr-system}.

$L$ is \emph{correct and complete} log of $S$ iff $L$ can be \emph{replayed} over the entire system $S$; see Def.~\ref{def:cpn_proclet_system:replay_semantics}. The event log of Tab.~\ref{tab:event_table_multiid} is a complete log of the PQR system in Fig.~\ref{fig:bhs}.

A correct and complete log $L$ has at least entity identifiers $pid$, $rid$, and $qid$ (as these are required by a PQR system). Further, each trace $\sigma_{et}^{id}$ can be replayed on the corresponding proclet, i.e., each trace describes a complete execution of the proclet for instance $id$. Further, all traces $\sigma_{pid}^{id} \in \sigma(L,pid)$ of process entities ($pid$) are ordered relative to each other via the shared resources and queues as described in $S$.

In reality often only a subset of activities $B \subseteq A = \{ \#_{act}(e) \mid e \in L \}$ and the control-flow identifier $pid$ have been recorded in a log, making it \emph{partial}.
\begin{definition}[Partial Log, Observed Event, Corresponds]\label{def:partial_log}
A log $L' = (E',\mathit{AN}',\{pid\},\#')$ is a \emph{partial} (and correct) log of PQR system $S$ if there exists a correct and complete log $L = (E,\mathit{AN},\mathit{ET},\#)$ of $S$ such that
\begin{enumerate}
  \item $E' \subseteq E$, $\mathit{AN}' \subseteq \mathit{AN}$, and $\#' = \#|_{E' \times \mathit{AN}'}$ is the restriction of $\#$ to $E'$ and $\mathit{AN}'$,
  \item each $e \in E'$ has only case notion $pid$, i.e., $\#_{pid}(e)  \neq \perp$ (and $\#_{rid}(e) = \#_{qid}(e) =\perp$),
  \item $\#_{time}(e)$ is defined, and
  \item for each complete process trace $\sigma(L,pid=id) = \langle e_1,\ldots,e_n\rangle$ the partial trace $\sigma(L',pid=id) = \langle f_1,\ldots,f_k\rangle$ records at least the first and last event $e_1 = f_1$ and $e_n = f_k$.
\end{enumerate}
We call each event $e \in E'$ an \emph{observed event}. We say that the complete log $L$ \emph{corresponds to} the partial log $L'$.
\end{definition}

\newcommand\cg[1]{ {\color{gray} #1} }

\begin{table}\centering
\caption{Partial event log of the complete event log in Tab.~\ref{tab:event_table_multiid}, missing events and attributes shown in grey}\label{tab:event_table_multiid:partial}
{\small\sffamily
\begin{tabular}{lrrrrr}
  \hline
  event id & pid & activity & time & \cg{rid} & \cg{qid} \\
  \hline
$e_0$ & 50 & $c3_c$ & 01.01.20 9:00:15 & \cg{$\perp$} & \cg{c3:m3}\\
$e_1$ & 50 & $m3_s$ & 01.01.20 9:00:30 & \cg{m3} & \cg{c3:m3}\\
\cg{$e_2$} & \cg{50} & \cg{$m3_c$} & \cg{01.01.20 9:00:40} & \cg{m3} & \cg{m3:m4}\\
\cg{$e_3$} & \cg{50} & \cg{$m4_s$} & \cg{01.01.20 9:00:45} & \cg{m4} & \cg{m3:m4}\\
\cg{$e_4$} & \cg{50} & \cg{$m4_c$} & \cg{01.01.20 9:00:50} & \cg{m4} & \cg{m4:d1}\\
$e_7$ & 50 & $d1_s$ & 01.01.20 9:01:05 & \cg{d1} & \cg{m4:d1}\\
\cg{$e_8$} & \cg{50} & \cg{$d1_c$} & \cg{01.01.20 9:01:10} & \cg{d1} & \cg{d1:s1}\\
$e_{18}$ & 50 & $s1_s$ & 01.01.20 9:01:15 & \cg{$\perp$} & \cg{d1:s1}\\
  \hline
$e_{17}$ & 51 & $c4_c$ & 01.01.20 9:00:35 & \cg{$\perp$} & \cg{c4:m4}\\
$e_5$ & 51 & $m4_s$ & 01.01.20 9:00:55 & \cg{m4} & \cg{c4:m4}\\
\cg{$e_6$} & \cg{51} & \cg{$m4_c$} & \cg{01.01.20 9:01:00} & \cg{m4} & \cg{m4:d1}\\
\cg{$e_9$} & \cg{51} & \cg{$d1_s$} & \cg{01.01.20 9:01:15} & \cg{d1} & \cg{m4:d1}\\
\cg{$e_{10}$} & \cg{51} & \cg{$d1_c$} & \cg{01.01.20 9:01:20} & \cg{d1} & \cg{d1:d2}\\
$e_{11}$ & 51 & $d2_s$ & 01.01.20 9:01:25 & \cg{d2} & \cg{d1:d2}\\
\cg{$e_{12}$} & \cg{51} & \cg{$d2_c$} & \cg{01.01.20 9:01:30} & \cg{d2} & \cg{d2:s2}\\
$e_{19}$ & 51 & $s2_s$ & 01.01.20 9:01:35 & \cg{$\perp$} & \cg{d2:s2}\\
\hline
\end{tabular}}
\end{table}

Thus, a partial log $L'$ contains for each case $pid$ at least one \emph{partial trace} $\sigma_{pid}^{id}$ recording the entry and exit of the case and preserving the order of observed events, i.e., it can be completed to fit the model. An MHS typically records a partial log as defined above. Tab.~\ref{tab:event_table_multiid:partial} shows a partial event log of the complete log of Tab.~\ref{tab:event_table_multiid}. Fig.~\ref{fig:example}(b) highlights the events that are recorded in the partial event log .

\medskip
Note that a partial event log coincides with the definition of a classical single-entity event log (Def.~\ref{def:event_log_single}). In a partial event log, events of different process cases are less ordered, e.g., observed events $e1$ and $e5$ in Fig.~\ref{fig:example} are unordered wrt.\ any resource or queue whereas they are ordered in the corresponding complete event log.

\begin{lemma}
Let $L'$ be a partial event log of a PQR system $S$. Let $L$ be a complete event log of $S$ that corresponds to $L'$. Let $\pi(L')$ and $\pi(L)$ be the system-level runs of $L'$ and $L$, respectively. Then for each $e_1,e_2 \in E'$: $e_1 < e_2$ in $L'$ implies $e_1 < e_2$ in $L$.
\end{lemma}

\begin{proof}
For any $e_1 < e_2$ in $L'$ holds $\#_{pid}(e_1) = \#_{pid}(e_2)$ and $\#_{time}(e_1) < \#_{time}(e_2)$. These properties also hold in $L$, thus $e_1 < e_2$ in $L$.
\end{proof}
The converse does not hold. In the complete system-level run $\pi(L)$ in Fig.~\ref{fig:example}(b), $e_5 < e_7$ (due to $qid = m4:d1$), whereas $e_5 \nless e_7$ in the system-level run of the partial log $L'$ (where $e_5$ and $e_7$ are unrelated). In the following, we investigate how to infer missing events and infer missing time-stamps, and thereby reconstruct the missing ordering relations.

\paragraph{Formal problem statement}

Let $S$ be a model of a PQR system defining life-cycles of process, resource, and queue proclets, which resources and queues synchronize on which process step, and for each resource the minimum service time $tsr$ and waiting time $twr$ and for each queue the minimum waiting time $twq$. Given $S$ and a partial event log $L_1$ of $S$, we want to construct a complete log $L_2$ of $S$ that corresponds to $L_1$ (and can be replayed on $S$) according to Def.~\ref{def:partial_log}.

\medskip
Restoring the \emph{exact} timestamp is generally infeasible and for most use cases also not required. Thus, our problem formulation does not require to reconstruct the exact time-stamps. Our CPN replay semantics allows to fire transitions after their first moment of enabling, however they have to fire ``early enough'' so that time constraints do not conflict with later observed events. Thus, we have to reconstruct \emph{time-windows} providing minimal and maximal timestamps for each unobserved event, resulting in the following sub-problems:
\begin{itemize}
  \item Infer unobserved events $E_u$ for all process cases in $L_1$ and their relations to queues and resources (infer missing identifiers)
  \item Infer for each unobserved event $e \in E_u$ a time-window of earliest and latest occurrence of the event $\#_{tmin}(e), \#_{tmax}(e)$ so that setting $\#_{time}(e) = \#_{tmin}(e)$ or $\#_{time}(e) = \#_{tmax}(e)$ for $e \in E_u$ results in a complete log of $S$.
\end{itemize}

\section{Inferring timestamps along entity traces}
\label{sec:solution}

In Sect.~\ref{sec:pqr:problem_statement}, we presented the problem of restoring missing events and time-windows for their timestamps from a partial event log $L_1 = (E_1,\mathit{AN}_1,\{pid\},\#^1)$ such that the resulting log is consistent with resource and queueing behavior specified in a PQR System $S$. In this section, we solve the problem for PQR systems with \emph{acyclic} process proclets by casting it into a constraint satisfaction problem, that can be solved using Linear Programming (LP)~\cite{Schrijver86}. In all subsequent arguments, we make extensive use of the fact that we can see any multi-entity event log $L_1$ equivalently as family of sequential event logs $\sigma(L_1,et)$ with traces $\sigma_{et}^{id}$ and as the system-level run $\pi(L_1) = (E_1,<_1,\mathit{AN}_1,\{pid\},\#^1)$ with strict partial order $(E_1,<_1)$.

\medskip
In Sect.~\ref{sec:approach}, we show how to infer unobserved events and resource and queue identifiers (from $S$) to construct an under-specified intermediate system-level run $\pi_2 = (E_2,<_2,\mathit{AN}_2,\{pid,rid,qid\},\#^2)$ where all unobserved events $E_u = E_2 \setminus E_1$ have no timestamp but where $<_2$ already contains all ordering constraint that must hold in $S$.

In Sect.~\ref{sec:algorithm} we then refine $\pi_2$ into $\pi(L_3) = (E_2,<_3,\mathit{AN}_3,\{pid,rid,qid\},\#^3)$ where $<_3$ is no longer explicitly constructed but completely inferred from time stamps that fit $S$. We determine minimal and maximal timestamps $\#_{tmin}^3(e)$ and $\#_{tmax}^3(e)$ for each unobserved event $e \in E_u$ (through a linear program) so that if we set $\#_{time}^3(e) = \#_{tmin}^3(e)$ or $\#_{time}^3(e) = \#_{tmax}^3(e)$, the induced partial order $<_3$ refines $<_2$, i.e., $\mathord{<_2} \subseteq \mathord{<_3}$. By construction of $\#_{tmin}^3(e)$ and $\#_{tmax}^3(e)$, $L_3$ is a complete log of $M$ and has $L_1$ as a partial log. We explain our approach using another (more compact running) example shown in Fig.~\ref{fig:oracle12a}(a) for two bags 53 and 54 processed in the system of Fig~\ref{fig:bhs}. The events in grey italic (i.e., f3, f5, f6, f14) are unobserved.

\begin{figure}[t!]
    \includegraphics[width=\linewidth]{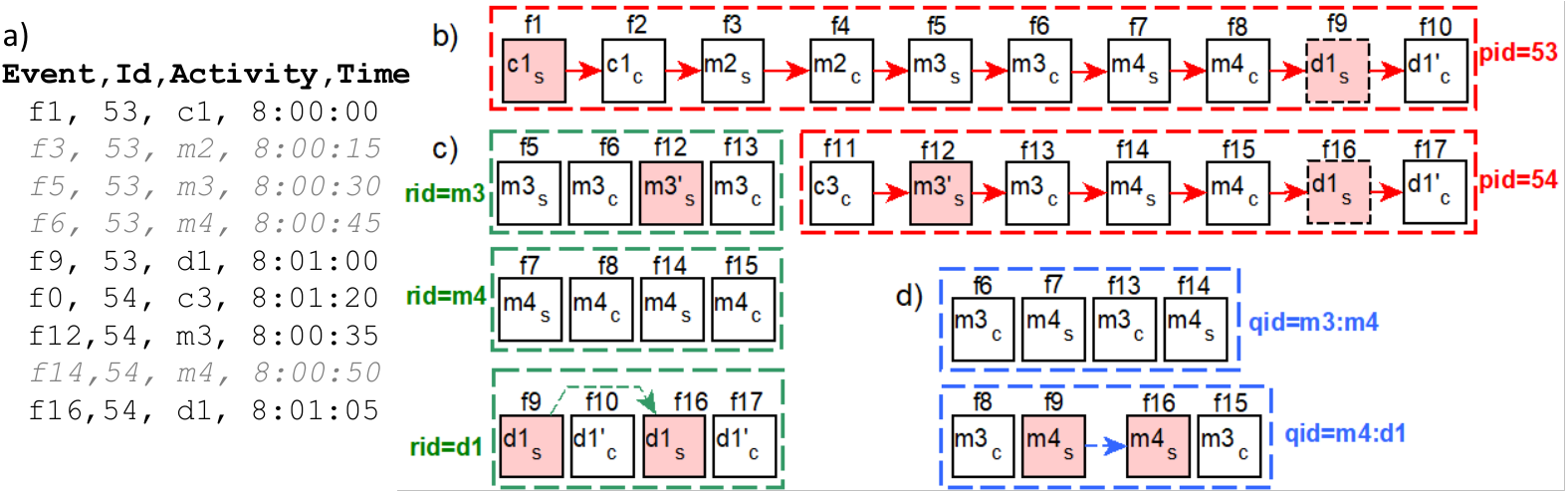}
    \caption{ Another partial event log of the system in Fig.~\ref{fig:bhs} for bags 53 and 54 (a), partially complete traces of the Process (b), Resource (c) and Queue (d) proclets, restored by oracles  $O_1, O_2$. Only observed events are ordered, e.g., $f9<_{rid}^{d1}f16$, while the other events are isolated.}
    \label{fig:oracle12a}
\end{figure}

\subsection{Infer potential complete runs from a partial run}
\label{sec:approach}

We first infer from the partial event log $L_1$ an under-specified intermediate system-level run $\pi_2$ containing all unobserved events and an explicitly constructed SPO $<_2$ so that each entity-level run $\pi_{2,pid}^{id}$ is complete (can be replayed on the process proclet in $S$). In a second step, we relate each unobserved event $e \in E_u = E_2 \setminus E_1$ to a corresponding resource and/or queue identifier which orders observed events wrt. $<_{rid}$ and $<_{qid}$. All unobserved events $e \in E_u$ lack a timestamp and hence are left unordered wrt.\ $<_{rid}$ and $<_{qid}$ in $\pi(L_2)$; we later refine $<_2$ in Sect~\ref{sec:algorithm}.

\medskip
We specify how to solve each of these two steps in terms of two \emph{oracles} $O_1$ and $O_2$ and describe concrete implementations for either.

\paragraph{Restoring process traces}

Oracle $O_1$ has to return a set of sequential traces $L_2 = \{ \sigma_{pid}^{id} \mid \mathit{id} \in \mathit{pid}(L_1) \} = O_1(L_1,S)$ by completing each partial trace $\sigma(L_1,pid=id)$ of any process case $id \in \mathit{pid}(L_1)$ into a complete trace $\sigma_{pid}^{id}$ that can be replayed on the process proclet of $S$. Let $E_2 = \{ e \in \sigma_{pid}^{id} \mid \sigma_{pid}^{id} \in L_2\}$. The restored \emph{unobserved} events $E_u = E_2 \setminus E_1$ only have attributes $act$ and $pid$ and events are \emph{totally ordered} along $pid$ in each trace $\sigma_{pid}^{id}$. $O_1$ can be implemented using well-known trace alignment~\cite{TraceAlginment} by aligning each sequential trace $\sigma(L_1,pid=id)$ on the skeleton net $(P,T,F)$ of the P-proclet of $S$. For example, applying $O_1$ on the partial log of Fig.~\ref{fig:oracle12a}(a) results in the complete process traces of Fig.~\ref{fig:oracle12a}(b).

At this point, the events $e \in E_u$ have no time-stamp and the ordering of events is only available in the explicit sequences $\sigma_{pid}^{id} = \langle e_1,\ldots,e_n \rangle$. Until we have determined $\#_{time}(e_i)$, the SPO $<_2$ has to be constructed explicitly from the ordering of events in the traces $\sigma_{pid}^{id}$, i.e., we define $<_2$ as $e_i < e_j$ iff there ex. a trace $\langle \ldots,e_i,\ldots,e_j,\ldots \rangle = \sigma_{pid}^{id} \in L_2$ (see Cor.~\ref{cor:sequence_po_equivalence}).

Moreover, as each trace $\sigma_{pid}^{id}$ can be replayed on the process proclet, each event is either a start event (replays a start transition $t \in T_{start}$) or a complete event (replays a complete transition $t \in T_{complete}$, see Def.~\ref{def:p-proclet}).

\paragraph{Inferring dependencies due to shared resources and queues}

Oracle $O_2$ has to enrich events in $E_2$ with information about queues and resources so that for each $e \in E_2$ if resource $r$ is involved in the step $\#_{act}(e)$, then $\#_{rid}(e) = r$ and if queue $q$ was involved, then $\#_{qid}(e) = q$.

Moreover, in order to formulate the linear program to derive timestamps in a uniform way, each event $e$ has to be annotated with the performance information of the involved resource and/or queue. That is, if $e$ is a start event and $\#_{rid}(e) = r \neq \perp$, then $\#_{tsr}(e)$ and $\#_{twr}(e)$ hold the minimum service and waiting time of $r$, and if $\#_{qid}(e) = q \neq \perp$, then $\#_{twq}(e)$ hold the minimum waiting time\linebreak of $q$.

For the concrete PQR systems considered in this paper, we set $\#_{rid}(e) = r$ based on the model $S$ if $r$ is the identifier of the resource proclet that synchronizes with transition $t = \#_{act}(e)$ via a channel $c(t)$ (there is at most one). Attributes $\#_{tsr}(e)$, $\#_{twr}(e)$, can be set from the model as they are parameters of the resource proclet. To ease the LP formulation, if $e$ is unrelated to a resource, we set $\#_{rid}(e) = r^*$ to fresh identifier and $\#_{tsr}(e) = \#_{twr}(e) = 0$; $\#_{qid}(e)$ and $\#_{twq}(e)$ are set correspondingly. By annotating the events in $E_2$ as stated above, we obtain $\pi_2 = (E_2,<_2,\mathit{AN}_2,\{pid,rid,qid\},\#^2)$. Moreover, we can update the SPO $<_2$ by inferring $\lessdot_{rid}$ and $\lessdot{qid}$ from $\#_{time}(e)$ for all events where $\#_{rid}(e) \neq \perp$ and $\#_{qid}(e) \neq \perp$ (see Def.~\ref{def:system-level_run}).

The system-level run $\pi_2$, contains complete entity-level runs for $pid$ (except for missing time stamps). The entity-level runs queues ($qid$) and resources ($rid$) already contain all events to be complete wrt.\ $S$ but only the observed events are ordered (due to their time stamps). For example, Fig.~\ref{fig:oracle12a}(d) shows the entity-level run $\pi_{qid}^\textit{m4:d1}$ containing events $f_8,f_9,f_{16},f_{15}$ with only $f_9 <_{qid} f_{16}$. Next, we define constraints based on the information in this intermediate run $\pi$ to infer timestamps for all unobserved events.

\subsection{Restoring timestamps of unobserved events by linear programming}
\label{sec:algorithm}

The SPO $\pi_2 = (E_2,<_2,\mathit{AN}_2,\{pid,rid,qid\},\#^2)$ obtained in Sect.~\ref{sec:approach} from partial log $L_1$ includes all unobserved events $E_u = E_2 \setminus E_1$ of the complete log, but lacks timestamps for each $e \in E_u, \#_{time}(e) = \perp$. Each observed $e \in E_1$ has a timestamp $\#_{time}(e)$ and we also added minimum service time $\#_{tsr}$, waiting time $\#_{twr}(e)$ of the resource $\#_{rid}(e)$ involved in $e$ and minimum waiting time $\#_{twq}(e)$ of the queue involved in $e$. We now  define a constraint satisfaction problem that specifies the earliest $\#_{tmin}(e)$ and latest $\#_{tmax}(e)$ timestamps for each $e \in E_u$ so that all earliest (latest) timestamps yield a consistent ordering of all events in $E$ wrt.\ $<_{pid}$ (events follow the process), $<_{rid}$ (events follow resource life-cycle), and $<_{qid}$ (events satisfy queueing behavior). The problem formulation propagates the known $\#_{time}(e)$ values along with the different case notions $<_{pid}$, $<_{rid}$, $<_{qid}$, using $tsr, twr, twq$.
For that, we introduce \textsf{variables} $x^{\mathrm{tmin}}_{e}$, $x^{\mathrm{tmax}}_{e} \geq 0$ for representing event attributes $tmin,tmax$ of each $e \in E_u$. For all observed events $e \in E_1$, we set $x_e^{\mathrm{tmin}} = x_e^{\mathrm{tmax}} = \#_{time}(e)$ as here the correct timestamp is known. We now define two groups of constraints to constrain the $x_e^{\mathrm{tmin}}$ and $x_e^{\mathrm{tmax}}$ values for the unobserved events further.

\medskip
In the following, we assume for the sake of simpler constraints presented in this paper, that all observed events are start events (which is in line with logging in an MHS). The constraints can easily be reformulated to assume only complete events were observed (as in most business process event logs) or a mix (requiring further case distinctions).

\subsubsection{Propagate information along process traces}

The first group propagates constraints for $\#_{time}(e)$ along $<_{pid}$, i.e., for each process-level run (viz. process trace) $\pi_{pid}^{id}$ of $pid$ in $\pi$. By the steps in Sect.~\ref{sec:approach}, events in $\pi_{pid}^{id}$ are totally ordered and derived from the trace $\sigma_{pid}^{id} = \langle e_1 ... e_m \rangle$.
Each process step has a start and a complete event in $\sigma_{pid}^{id}$, i.e., $m = 2\cdot y, y\in \mathbb{N}$, odd events are start events and even events are complete events.
For each process step $1 \leq i \leq y$, the time between start event $e_{2i-1}$ and complete event $e_{2i}$ is at least the service time of the resource involved (which we stored as $\#_{tsr}(e_{2i-1})$ in Sect.~\ref{sec:approach}). Thus the following constraints must hold for the earliest and latest time of $e_{2i-1}$ and $e_{2i}$.
\begin{equation}\label{eq:eq_tsr_min}
x^{\mathrm{tmin}}_{e_{2i}} = x^{\mathrm{tmin}}_{e_{2i-1}} + \#_{tsr}(e_{2i-1}) ,
\end{equation}
\begin{equation}\label{eq:eq_tsr_max}
x^{\mathrm{tmax}}_{e_{2i}}= x^{\mathrm{tmax}}_{e_{2i-1}} + \#_{tsr}(e_{2i-1}).
\end{equation}
For the remainder, it suffices to formulate constraints only for $start$ events.
We make sure that $tmin$ and $tmax$ define a proper interval for each start event:
\begin{equation}\label{eq:eq_tsr_min_max}
x^{\mathrm{tmin}}_{e_{2i-1}} \leq x^{\mathrm{tmax}}_{e_{2i-1}}.
\end{equation}
We write $e_i^s = e_{2i-1}$ for the start event of the i-th process step in $\sigma_{pid}^{id}$ and write $\theta_{pid}^{id} = \langle e_1^s, ..., e_m^s \rangle$ for the sub-trace of start events of $\sigma_{pid}^{id}$. Any event $e_i^s \in \theta_{pid}^{id}$ that was observed in $L_1$, i.e., $e_i^s \in E_1$, has $\#_{time}(e_i^s) \neq \perp$ defined. By Def.~\ref{def:partial_log}, $\sigma_{pid}^{id}$ as well as $\theta_{pid}^{id}$ always start and end with observed events, i.e., $e_1^s, e_y^s \in E_1$ and $\#_{time}(e_1^s), \#_{time}(e_y^s) \neq \perp$.
An unobserved event $e_i^s$ has no timestamp $\#_{time}(e_i^s) = \perp$ yet, but $\#_{time}(e_i^s)$ is bounded by $\#_{time}(e_1^s)$ (minimally) and $\#_{time}(e_y^s)$ (maximally).
Furthermore, any two succeeding start events in $\theta_{pid} = \langle ..., e_{i-1}^s, e_i^s,... \rangle$ are separated by the service time $\#_{tsr}(e_{i-1}^s)$ of step $e_{i-1}^s$ and the waiting time $\#_{twq}(e_i)$ of the queue from $e_{i-1}$ to $e_i$. Similar to Eq.~\ref{eq:eq_tsr_min} and \ref{eq:eq_tsr_max}, we formulate this constraint for both $x^{\mathrm{tmin}}_e$ and $x^{\mathrm{tmax}}_e$ variables:
\begin{equation}
x^{\mathrm{tmin}}_{e_k^s} \geq  x^{\mathrm{tmin}}_{e_{k-1}^s} + (\#_{tsr}(e_{k-1}^s)+\#_{twq}(e_{k}^s)), \numberthis \label{eq:eq_group1_tmin}
\end{equation}
\begin{equation}
x^{\mathrm{tmax}}_{e_k^s} \leq  x^{\mathrm{tmax}}_{e_{k+1}^s} - (\#_{tsr}(e_{k}^s)+\#_{twq}(e_{k+1}^s)). \numberthis \label{eq:eq_group1_tmax}
\end{equation}
Fig.~\ref{fig:ps_algorithm} uses the \emph{Performance Spectrum}~\cite{DFA_unbiased_fine_grained} to illustrate the effect of applying our approach step by step to the partially complete traces of Fig.~\ref{fig:oracle12a} obtained in the steps of Sect.~\ref{sec:approach}.
The straight lines in Fig.~\ref{fig:ps_algorithm}(a) from $f_1$ to $f_9$ (for pid=53) and from $f_{12}$ to $f_{16}$ (for pid=54) illustrate that $L_2$ (after applying $O_1$) contains all intermediate steps that both process cases passed through but not their timestamps. Further (after applying $O_2$), we know for each process step the resources (i.e., c1, m2, m3, m4, d1) and the queues (c1:m2, m2:m3 etc.), and their minimum service and waiting times $tsr, twr, twq$. The sum $tsr+twq$ is visualized as bars on the time axis in Fig.~\ref{fig:ps_algorithm}(a), the duration of $twr$ is shown in Fig.~\ref{fig:ps_algorithm}(b).

\medskip
We now explain the effect of applying Eq.~\ref{eq:eq_group1_tmin} on pid=53 for $f_3$, $f_5$ and $f_7$.
We have $\theta_{pid}^{53} = \langle f_1, f_3, f_5, f_7, f_9 \rangle$ with $f_1$ and $f_9$ observed, thus $x_{f_i}^{\mathrm{tmin}} = x_{f_i}^{\mathrm{tmax}} = \#_{time}(f_i)$ for $i \in \{1,9\}$. By Eq.~\ref{eq:eq_group1_tmin}, we obtain the lower-bound for the time for $f_3$ by $x_{f_3}^{\mathrm{tmin}} \geq x_{f_1}^{\mathrm{tmin}} + \#_{tsr}(f_1) + \#_{twq}(f_3)$ with $\#_{tsr}(f_1)$ and $\#_{twq}(f_3)$ the service time of resource c1 and waiting time of queue c1:m2. Similarly, Eq.~\ref{eq:eq_group1_tmin} gives the lower bound for $f_5$ from the lower bound from $f_3$ etc. Conversely, the upper bounds $x_{f_i}^{\mathrm{tmax}}$ are derived from $f_9$ ``downwards'' by Eq.~\ref{eq:eq_group1_tmax}.
This way, we obtain for each $f_i \in \theta_{pid}^{53}$ an initial interval for the time of $f_i$ between the bounds $x_{f_i}^{\mathrm{tmin}} \leq x_{f_i}^{\mathrm{tmax}}$ as shown by the intervals in Fig.~\ref{fig:ps_algorithm}(a).
As $x_{f_1}^{\mathrm{tmin}} = x_{f_1}^{\mathrm{tmax}} = \#_{time}(f_1)$ and $x_{f_1}^{\mathrm{tmin}} = x_{f_1}^{\mathrm{tmax}} = \#_{time}(f_9)$, the lower and upper bounds for the unobserved events in $\theta_{pid}^{53}$ form a polygon as shown in Fig.~\ref{fig:ps_algorithm}(b). Case 53 must have passed over the process steps and resources as a path inside this polygon, i.e., the polygon contains all admissible solutions for the timestamps of the unobserved events of $\theta_{pid}^{53}$; we call this polygon the \emph{region} of case 53. The region for case 54 overlays with the region for case 53.

\begin{figure}[t!]
\vspace*{-1mm}
\scalebox{0.99}{
    \includegraphics[width=\linewidth]{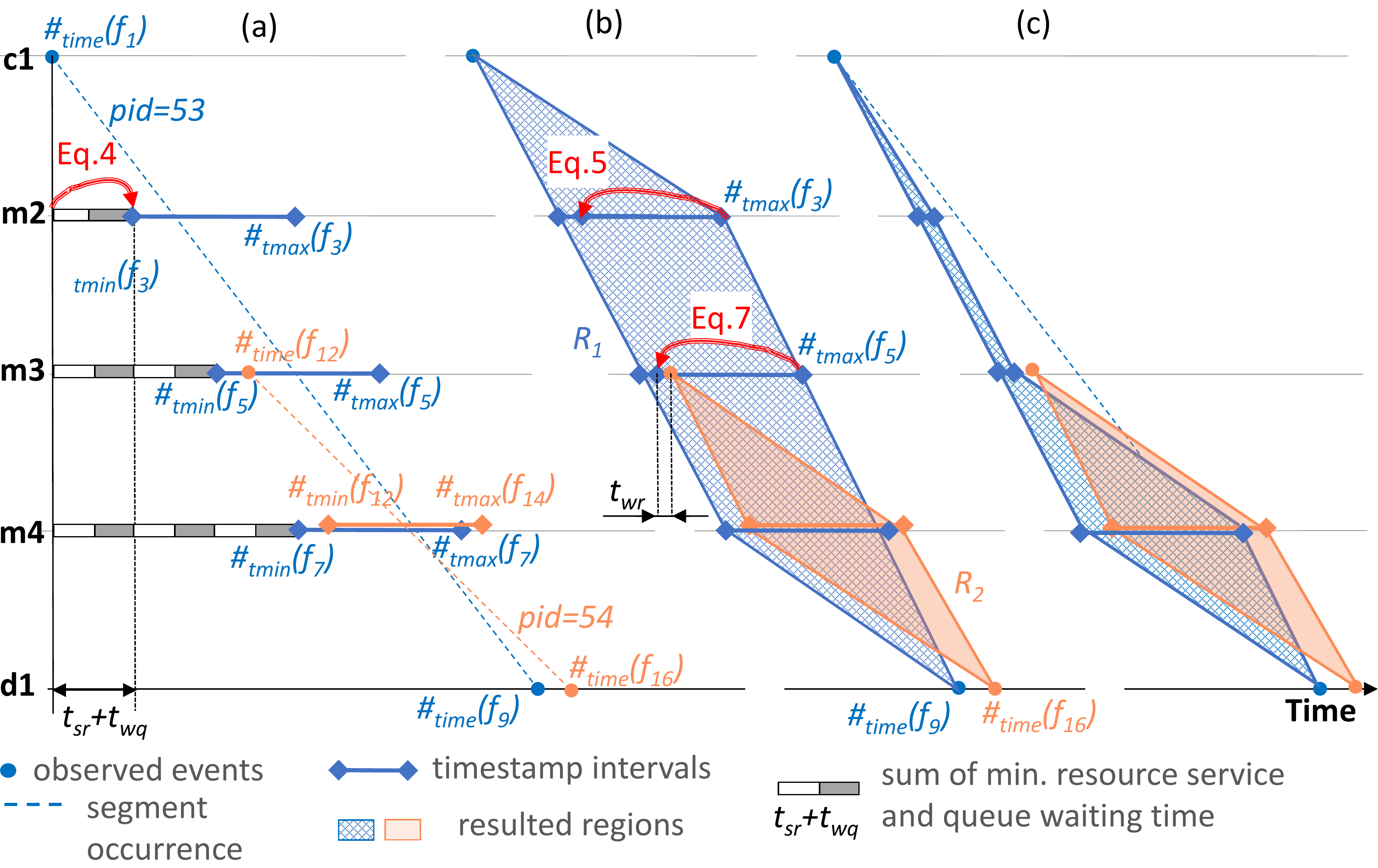}}    \vspace*{-5mm}
    \caption{Equations \ref{eq:eq_tsr_min}-\ref{eq:eq_group1_tmax} define time intervals for unobserved events (a), defining regions for the possible traces (b). Equations \ref{eq:eq_group2_tmin}-\ref{eq:eq_group2_tmax} propagate orders of cases observed on one resource to other resources (b), resulting in tighter regions (c).}    \label{fig:ps_algorithm}\vspace*{-1mm}
\end{figure}

\subsubsection{Propagate information along resource traces}

We now introduce a second group of constraints by which we infer more tight bounds for $x_{e_i}^{\mathrm{tmin}}$ and $x_{e_i}^{\mathrm{tmax}}$ based on the overlap with other regions. While the first group of constraints traversed entity traces along $pid$ (i.e., process traces), the second group of constraints traverses entity traces for resources along $rid$.

\medskip
Each resource trace $\pi_{rid}^r$ in $\pi_2$, contains all events $E_{rid}^r$ resource $r$ was involved in\,--\,\emph{across} multiple different process traces. The SPO $<_{rid}^r$ orders \emph{observed} events of this resource trace due to their known timestamps; e.g. in Fig.~\ref{fig:ps_algorithm}(b) $f_9 <_{rid}^{m1} f_{16}$ with $f_9$ from pid=53 and $f_{16}$ from pid=54.

The order of the two events $e_{p1}^s <_{rid}^r e_{p2}^s$ for the \emph{same} step $\#_{act}(e_{p1}^s) = \#_{act}(e_{p2}^s) = t_1$ in different cases $\#_{pid}(e_{p1}^s) = p1 \neq \#_{pid}(e_{p2}^s) = p2$ propagates ``upwards'' and ``downwards'' the process traces $\pi_{pid}^{p1}$ and $\pi_{pid}^{p2}$ as follows. Let events $f_{p1}^s \in E_{pid}^{p1}$ and $f_{p2} \in E_{pid}^{p2}$ be events in process traces $\pi_{pid}^{p1}$ and  $\pi_{pid}^{p2}$ of the same step $\#_{act}(f_{p1}^s) = \#_{act}(f_{p2}^s) = t_n$. We say $t_1$ and $t_n$ are \emph{in FIFO relation} iff  there is a unique path $\langle t_1 ... t_n \rangle$ between $t_1$ and $t_n$ in the process proclet (i.e., no loops, splits, parallelism) so that between any two consecutive transitions $t_k$, $t_{k+1}$ only synchronize with single-server resources or FIFO queues. If $t_1$ and $t_n$ are in FIFO relation, then also $f_{p1}^s <_{rid}^{r2} f_{p2}^s$ on the resource r2 involved in $t_n$ (as the case $p1$ cannot overtake the case $p2$ along this path). Thus $x_{f_{p1}^s}^{\mathrm{tmin}} \leq x_{f_{p2}^s}^{\mathrm{tmin}}$ must hold. More specifically, $x_{f_{p1}^s}^{\mathrm{tmin}}  + \#_{tsr}(f_{p1}^s) + \#_{twr}(f_{p1}^s) \leq x_{f_{p2}^s}^{\mathrm{tmin}}$ must hold as the service time and waiting time of the resource involved in $f_{p1}^s$ must elapse.

\medskip
For any pair $e_{p1}^s, e_{p2}^s \in E_{rid}^r$ with $e_{p1}^s <_{rid}^r e_{p2}^s$ and any other trace $\theta_{rid}^{r2}$ for resource $r2$ and any pair $f_{p1}^s, f_{p2}^s \in E_{rid}^{r2}$ such that $\#_{pid}(e_{p1}^s) = \#_{pid}(f_{p1}^s), \#_{pid}(e_{p2}^s) = \#_{pid}(f_{p2}^s)$ and transition $\#_{act}(e_{p1}^s)$ is in FIFO relation with $\#_{act}(f_{p1}^s)$, we generate the following constraint for $tmin$ between different process cases $p1$ and $p2$:
\begin{equation}\label{eq:eq_group2_tmin}
x^{\mathrm{tmin}}_{f_{p1}^s}  \leq x^{\mathrm{tmin}}_{f_{p2}^s}-(\#_{tsr}(f_{p1}^s)+\#_{twr}(f_{p1}^s)),
\end{equation}
and the following constraint for $tmax$:
\begin{equation}\label{eq:eq_group2_tmax}
x^{\mathrm{tmax}}_{f_{p1}^s}  \leq x^{\mathrm{tmax}}_{f_{p2}^s}-(\#_{tsr}(f_{p1}^s)+\#_{twr}(f_{p1}^s)),
\end{equation}
In the example of Fig.~\ref{fig:ps_algorithm}(b), we observe $f_9 <_{rid}^{d1} f_{16}$ (both of transition $d1_s$) along resource $d1$ at the bottom of Fig.~\ref{fig:ps_algorithm}(b).
By Fig.~\ref{fig:bhs}, $d1_s$ and $m3_s$ are in FIFO-relation. Applying Eq.~\ref{eq:eq_group2_tmax} yields $x_{f_5}^{\mathrm{tmax}} \leq \#_{time}(f_{12}) - (\#_{tsr}(f_5)+\#_{twr}(f_5))$, i.e., $f_5$ occurs at latest before $f_{12}$ minus the service and waiting time of $m3$.
This operation significantly reduces the initial region $R_1$.
By Eq.~\ref{eq:eq_group1_tmax}, the tighter upper bound for $f_5$ also propagates along the trace pid=53 to $f_3$, i.e., $x_{f_3}^{\mathrm{tmax}} \leq x_{f_5}^{\mathrm{tmax}} - (\#_{tsr}(f_3) + \#_{twq}(f_5))$, resulting in a tighter region as shown in Fig.~\ref{fig:ps_algorithm}(c).
If another trace $\langle m3_s, d1_s \rangle$ were present \emph{before} trace $53$, then this would cause reducing the $tmin$ attributes of the events of trace $53$ by Eq.~\ref{eq:eq_group1_tmin},\ref{eq:eq_group2_tmin} in a similar way.
In general, the more cases interact through shared resources, the more accurate timestamp intervals can be restored by Eq.~\ref{eq:eq_tsr_min}-\ref{eq:eq_group2_tmax} as we will show in Sect.~\ref{sec:evaluation}.

\medskip
To construct the linear program, we generate equations \ref{eq:eq_tsr_min} to \ref{eq:eq_group1_tmax} by iteration of each process trace in $L_2$. Further, iterate over each resource trace and for each pair of events $e_{p1} <_{rid}^{r} e_{p2}$ we generate equations \ref{eq:eq_group2_tmin},\ref{eq:eq_group2_tmax} for each other pair of events $f_{p1} <_{rid}^{r2} f_{p2}$ that is in FIFO relation. The objective function to maximize is the sum of all intervals $\sum_{e \in E_2}(x_e^{\mathrm{tmax}} - x_e^{\mathrm{tmin}})$ to maximize the coverage of possible time-stamp values by those intervals.

Solving this linear program assigns to each event $e \in E_2$ upper and lower bounds $\#_{tmin}(e)$ and $\#_{tmax}(e)$ for $\#_{time}(e)$; $\#_{tmin}(e) = \#_{time}(e) = \#_{tmax}(e)$ for all $e \in E_1$ (by \ref{eq:eq_tsr_min} and \ref{eq:eq_tsr_max} the solutions for the start events propagate to complete events with time difference $tsr$). By setting $\#_{time}(e) = \#_{tmin}(e)$ (or $\#_{time}(e) = \#_{tmax}(e)$) we obtain $L_3 = (E_2,\mathit{AN}_3,\{pid,rid,qid\},\#^3)$ where the SPO $<_3$ of the system-level run $\pi(L_3)$ refines the SPO $<_2$ constructed explicitly in Sect.~\ref{sec:approach}.

By oracle $O_1$, $\sigma(L_3,pid)$ can be replayed on the P-proclet.

By \ref{eq:eq_tsr_min} and \ref{eq:eq_group2_tmin}, for any two events $e \lessdot_{rid} e'$ the time difference is $\#_{time}(e') - \#_{time}(e) > twr$ or $\#_{time}(e') - \#_{time}(e) > tsr$ of the corresponding R-proclet $R_i$ (depending on whether $e$ replays by the start or the complete transition of $R_i$). Thus, $\sigma(L_3,r)$ can be replayed on the correspond R-proclet for any resource $r$ in $L_3$.

By \ref{eq:eq_tsr_min} and \ref{eq:eq_group1_tmin}, the time-stamps of $e \lessdot_{pid} e'$ where $e$ replays $t_{enq}$ and $e'$ replays $t_{deq}$ of a Q-proclet $Q_i$ have at least time difference $twq$ of $Q_i$ (i.e., the time constraint of $Q_i$ is satisfied). If for two process cases $p1$ and $p2$ we observe $e_{p1} <_{rid} e_{p2}$ at the same step $\#_{act}(e_{p1}) = \#_{act}(e_{p2})$ with $\#_{pid}(e_1) = p1 \neq p2 = \#_{pid}(e_2)$ at some step, the we also observe $f_{p1} <_{rid} f_{p2}$ at another step $\#_{act}(f_{p1}) = \#_{act}(f_{p2})$ with $\#_{pid}(f_1) = p1 \neq p2 = \#_{pid}(f_2)$ at later events $e_1 <_{pid} f_1$ and $e_2 <_{pid} f_2$ (by \ref{eq:eq_group2_tmin} and \ref{eq:eq_group2_tmax}). As in a PQR system, for each queue, the enqueue transition synchronizes with a different resources than the dequeue transition, the relation $e_{p1} <_{qid} e_{p2}$ and $f_{p1} <_{qid} f_{p2}$ also holds if $e_{p1}, e_{p2}$ are enqueue events and $f_{p1}, f_{p2}$ are dequeue events of the same queue $Q_i$. Thus the FIFO constraint of $Q_i$ is satisfied. Thus, $\sigma(L_3,q)$ can be replayed on the correspond Q-proclet for any queue $q$ in $L_3$.

Altogether, $L_3$ is a complete log that can be replayed on the PQR system $S$ (by Def.~\ref{def:pqr-system} and Def.~\ref{def:cpn_proclet_system:replay_semantics}).\vspace*{-1mm}

\section{Evaluation}
\label{sec:evaluation}

To evaluate our approach, we formulated the following questions.
(Q1)  Can timestamps be estimated in real-life settings and used to estimate performance reliably?
(Q2) How accurately can the load (items per minute) be estimated for different system parts, using restored timestamps?
(Q3) What is the impact of sudden deviations from the minimum service/waiting times, e.g., the unavailability of resource or stop/restart of an MHS conveyor, on the accuracy of restored timestamps and the computed load?
For that, we extended the interactive ProM plug-in ``Performance Spectrum Miner'' with an implementation of our approach that solves the constraints using heuristics\footnote{The simulation model, simulation logs, ProM plugin, and high-resolution figures are available on \url{https://github.com/processmining-in-logistics/psm/tree/rel}.}.
As input we considered the process of a part of real-life BHS shown in Fig.~\ref{fig:sim_mfd} and used Synthetic Logs (SL) (simulated from a model to obtain ground-truth timestamps) and Real-life Logs (RL) from a major European airport. Regarding Q3, we generated SL with regular performance and with \emph{blockages} of belts (i.e., a temporary stand-still); the

\begin{figure}[!h]
\vspace*{1mm}
    \includegraphics[width=\linewidth]{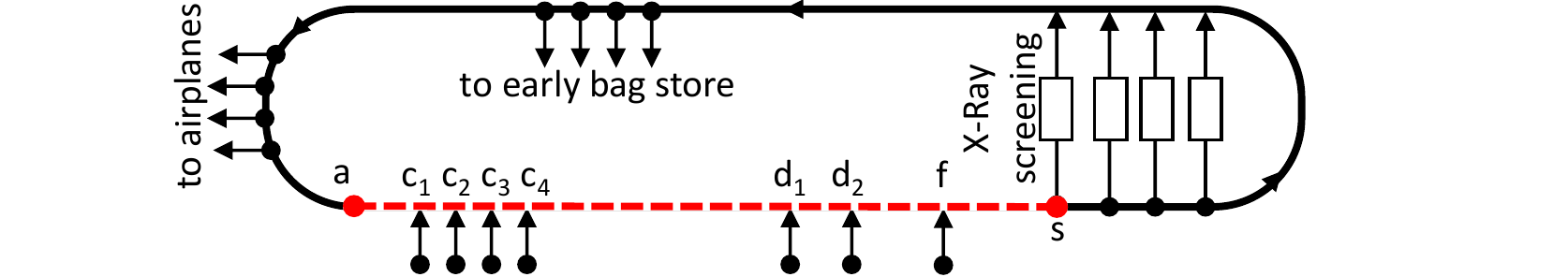}
    \caption{In the BHS bags come from check-in counters $c_{1-4}$ and another terminals $d_{1-2},f$, go through mandatory screening and continue to other locations.}
    \label{fig:sim_mfd}
\end{figure}

\noindent RL contained both performance characteristics. All logs were partial as described in Sect.~\ref{sec:problem_statement}. We selected the acyclic fragment highlighted in Fig.~\ref{fig:sim_mfd} for restoring timestamps of steps $c_{1-4},d_{1-2},f,s$. 

\begin{figure}[h]
\vspace*{3mm}
\centering
    \includegraphics[width=.59\linewidth]{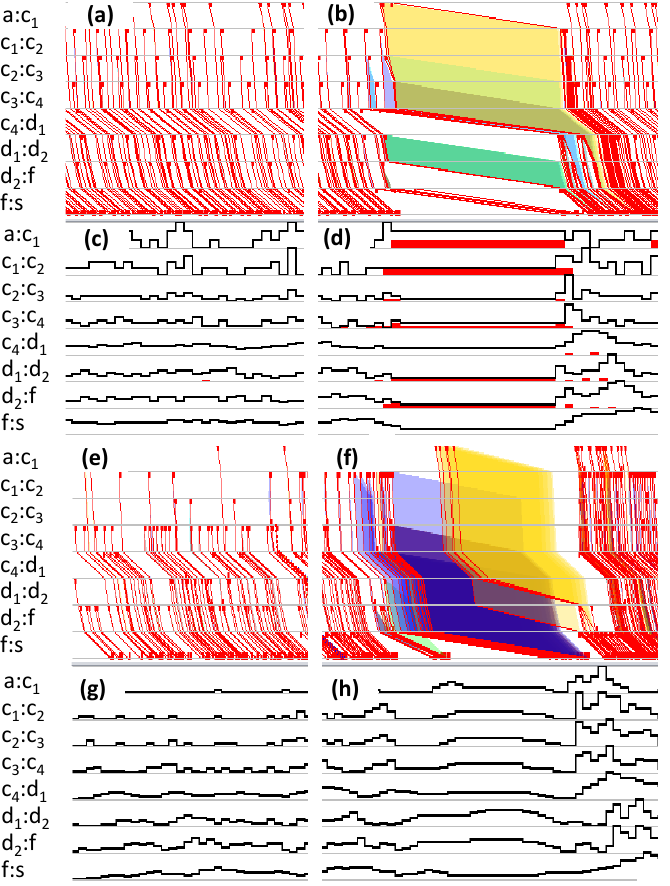}
    \caption{Restored Performance Spectrum for synthetic (a,b) and real-life (e,f) logs. The estimated load (computed on estimated timestamps) for synthetic (c,d) and real-life (g,h) logs. For the synthetic logs, the load error is measured and shown in red (c,d). Colored-shaded regions indicate for selected traces the boundaries of timestamps of reconstructed events between different observed events $a$ to $s$ (yellow), $c1$ to $s$ (blue), $d1$ to $s$ (green).
    }
    \label{fig:ps_experiments}
\end{figure}

We evaluated our technique against the ground truth known for SL as follows. For each event we measured the error of the estimated timestamp intervals $[t_{min},t_{max}]$ against the actual time $t$ as $\max \{ |t_{max}-t|,|t_{min}-t| \}$ normalized over the sum of minimal service and waiting times of all involved steps (to make errors comparable). We report the Mean Absolute Error (MAE) and Root Mean Square Error (RMSE) of these errors. Applying our technique to SL with regular behavior, we observed very narrow time intervals for the estimated timestamps, shown in Fig.~\ref{fig:ps_experiments}(a), and a MAE of $<5\%$. The MAE of the estimated load (computed on estimated timestamps), shown in Fig.~\ref{fig:ps_experiments}(c), was $<2\%$. For SL with blockage behavior, the intervals grew proportionally with the duration of blockages (Fig.~\ref{fig:ps_experiments}(b)), leading to a proportional growth of the MAE for the timestamps. However, the MAE of the estimated load (Fig.~\ref{fig:ps_experiments}(d)) was at most $4\%$. The load MAE for different processing steps for both scenarios are shown in Table~\ref{table:results}.
Notably, both observed and reconstructed load showed load peaks each time the conveyor belt starts moving again.

\begin{table}[h]
\vspace*{-1mm}
\begin{center}
\caption{The estimated load (computed on estimated timestamps) Root Mean Squared Error (RMSE) and Mean Absolute Error (MAE) are shown in \% of max. load.}\label{table:results}
\scalebox{0.98}{
\begin{tabular}[t!]{ lrrr }
 \hline
 Scenario &   $c_4:d_1$ & $d_1:d_2$ & $f:s$ \\
 \hline
   no blockages\\
   \hspace{1em}MAE  & 0.16 &  0.22 &  0.17 \\
   \hspace{1em}RMSE & 1.01 & 1.66  &  0.89 \\
 \hline
   blockages\\
   \hspace{1em}MAE  &1.67  & 3.19  & 0.15 \\
   \hspace{1em}RMSE & 4.8  & 7.17  & 0.75 \\
  \hline
\end{tabular} }
\end{center}\vspace*{-1mm}
\end{table}

When evaluating on the real-life event log, we measured errors of timestamps estimation as the length of the estimated intervals (normalized over the sum of minimal service and waiting times of all involved steps). Performance spectra built using the restored RL logs are shown in Fig.~\ref{fig:ps_experiments}(e,f), and the load computed using these logs is shown in Fig.~\ref{fig:ps_experiments}(g,h).
The observed MAE was $<5\%$ in regular behavior and increased proportionally as observed on SL. The load error could not be measured, but similarly to synthetic data, it showed peaks after assumed conveyor stops.

The obtained results on SL show that the timestamps can be always estimated, and the actual timestamps are always within the timestamp intervals (Q1). When the system resources and queues operate close to the known performance parameters $tsr, twr, twq$, our approach restores accurate timestamps resulting in reliable load estimates in SL (Q2).
During deviations in resource performance, the errors increase proportionally with performance deviation while the estimated load remains reliable (error $<4\%$ in SL) and shows known characteristics from real-life systems on SL and RL (Q3).

\section{Conclusion}
\label{sec:conclusion}

In this paper, we studied the problem of repairing a partial event log with missing events for the performance analysis of systems where case interact and compete for shared limited resources.
We addressed the problem of repairing partial event logs that contain only a subset of events which impede the performance analysis of systems with shared limited resources and queues.

To study and solve the problem, we had to develop novel syntactic and semantic models for behavior over multiple entities. We specifically introduced a generalized model of event data over multiple behavioral entities that can be viewed, both, as sequential traces (with shared events) and as a partial order over the entire system behavior. We have shown in solving our problem that both perspectives are needed when reasoning about behavior of multiple entities.

To express domain knowledge about resources and queues, we had to extend the model of synchronous proclets~\cite{DFahlandMultiDim} with concepts for time and data, resulting in the notion of CPN proclet systems. A side effect of our work is a replay semantics for regular coloured Petri nets. We defined the sub-class of PQR systems to model processes served by shared resources and queues. Our model allows to decompose the interactions of resources and queues over multiple process cases into entity traces for process cases, resources and queues that synchronize on shared events (both on the syntactic and on the semantic level).

We exploit the decomposition when restoring missing events along the process traces using trace alignment~\cite{Carmona_ConfChecking}. We exploit the synchronization when formulating linear programming constraints over timestamps of restored events along, both, the process and the resource traces. As a result, we obtain timestamps which are consistent for all events along the process, resource, and queue dimensions. The evaluation of our implementation on synthetic and real-life data shows errors of the estimated timestamps and of derived performance characteristics (i.e., load) of $<5\%$ under regular performance, while correctly restoring real-life dynamics (i.e. load peaks) after irregular performance behavior.

\vspace{-2mm}
\paragraph{Limitations}
The work made several limiting assumptions. (1) Although the proclet formalism allows for arbitrary, dynamic synchronizations between process steps, resources, and queues, we limited ourselves in this work to a static known resource/queue id per process step. The limitation is not severe for some use cases such as analyzing MHS, but generalizing oracle $O_2$ to a dynamic setting is an open problem. (2) The LP constraints to restore timestamps assume an acyclic process proclet without concurrency. Further, the LP constraints assume 1:1 interactions (at most one resource and/or queue per process step). Both assumptions do not hold in business processes in general; formulating the constraints for a more general setting is an open problem. (3) Our approach ensures consistency of either all earliest or all latest timestamps with the given model, it does not suggest how to select timestamps between the latest and
earliest such that the consistency holds. (4) When the system performance significantly changes, e.g., due to sudden unavailability of resources, the error of restored timestamps is growing proportionally with the duration of deviations. Points (3) and (4) require attention to further improve event log quality for performance analysis.

\vspace{-2mm}
\paragraph{Future work}
Besides addressing the above limitations, our novel syntactic and semantic models open up new alleys of research for modeling and analyzing behavior over multiple entities, including more general conformance and process discovery. Moreover, the replay semantics for coloured Petri nets is likely to enable new kinds of process mining and conformance checking analyses beyond the types of systems studied in this paper.

\vspace{-2mm}
\paragraph{Acknowledgements} The research leading to these results has received funding from
Vanderlande Industries in the project ``Process Mining in Logistics''. We also thank Mitchel Brunings for his comments that greatly improved our approach.



\end{document}